\documentclass[11pt,reqno]{amsart}

\usepackage[margin=1in]{geometry}
\usepackage{setspace}
\usepackage{amsmath,amssymb,amsfonts,mathtools,bm}
\usepackage{amsthm}
\usepackage{bbm}
\usepackage{dsfont}

\usepackage{graphicx}
\usepackage{subcaption}
\usepackage{booktabs}
\usepackage{array}
\usepackage{multirow}
\usepackage{longtable}
\usepackage{float}

\usepackage[ruled,vlined]{algorithm2e}

\usepackage[numbers,sort&compress]{natbib}
\usepackage[colorlinks=true,linkcolor=blue,citecolor=blue,urlcolor=blue]{hyperref}

\newtheorem{theorem}{Theorem}[section]
\newtheorem{proposition}[theorem]{Proposition}

\theoremstyle{definition}
\newtheorem{definition}[theorem]{Definition}
\newtheorem{remark}[theorem]{Remark}

\usepackage{placeins}

\newcommand{\R}{\mathbb{R}}

\newcommand{\dd}{\,\mathrm{d}}

\usepackage{titlesec}

\titleformat{\section}
 [block]
 {\normalfont\Large\bfseries}
 {\thesection.}
 {0.75em}
 {}

\titleformat{\subsection}
 [block]
 {\normalfont\bfseries}
 {\thesubsection.}
 {0.75em}
 {}

\titleformat{\subsubsection}
 [block]
 {\normalfont\itshape}
 {\thesubsubsection.}
 {0.75em}
 {}

\newcommand{\papertitle}{Surface-Driven Stochastic Volatility for Commodity Options: Identification of Stochastic Vol-of-Vol and Leverage from Smile Dynamics}
\begin{document}

\begin{center}
{\Large\bfseries \papertitle\par}
\vspace{0.8em}
{\large Arthur Steve Tchoneteck\textsuperscript{1}, Tingjia Zhang\textsuperscript{2}, and Frederi Viens\textsuperscript{1}\par}
\vspace{0.5em}
{\small
\textsuperscript{1}Department of Statistics, Rice University, Houston, TX 77005, USA\\
\textsuperscript{2}ENSTA Paris, Paris, France\\
\texttt{st83@rice.edu}; \texttt{tingjia.zhang@ensta.fr}; \texttt{fv15@rice.edu}
\par}
\vspace{1em}
\end{center}
\vspace{0.5em}

\begin{center}
\begin{minipage}{0.92\textwidth}
\small
\noindent\textbf{Abstract. }Commodity option surfaces contain information beyond the at-the-money volatility level. We develop a surface-driven stochastic-volatility framework for soybean futures options using daily  Chicago Mercantile Exchange Group Volatility Index (CME CVOL) indicators from October 2013 to August 2025. The ATM level and convexity are positive, right-skewed, and well described by log-Ornstein--Uhlenbeck dynamics, whereas the additive skew changes sign and is modeled on the real line by an Ornstein--Uhlenbeck process. These empirical features motivate a joint mean-reverting surface-factor system. Under a log-OU futures-volatility model, leading-order smile relations map the ATM level to the latent volatility state, convexity to an effective surface-implied volatility-of-volatility state, and skew to an effective leverage state. We derive the associated risk-neutral pricing equation and Feynman--Kac representation, benchmark prices by Monte Carlo simulation, validate the pricing PDE by finite differences, and train a neural-network surrogate for fast repeated valuation. The empirical results show substantial variation in surface states across market regimes, with the 2021 La Niña drought window exhibiting elevated volatility and skewness. The numerical results show close agreement between finite-difference and Monte Carlo prices over the validation grid, while the neural-network surrogate reproduces the pricing map at substantially lower inference cost. Overall, the results support the use of level, skew, and convexity jointly as dynamic inputs for commodity-option pricing and model diagnostics.
\end{minipage}
\end{center}

\small
\noindent \textbf{Keywords:} commodity options; implied volatility surface; stochastic volatility; log-OU process; volatility-of-volatility; leverage; CME CVOL; soybean futures.\\
\textbf{JEL classification:} G13, C58, C52, Q14.\\
\textbf{Data:} CME CVOL soybean daily surface indicators, October 2013--August 2025.

\vspace{0.5em}

\section{Introduction}
\label{sec:introduction}

\subsection{Background and motivation}
Commodity option markets provide a useful yet demanding laboratory for stochastic volatility modeling. Unlike equity index options, agricultural and energy options are tied to physical production, storage, transportation, weather, and calendar-specific information. For soybean futures, the option surface changes around planting, pollination, harvest reports, export news, global supply shocks, and predictable exceptional climate phases like the La~Ni\~na and El~Ni\~no phases of the  El~Ni\~no Southern Oscillation (ENSO) global climate mode\footnote{It is common, though not a rule, for a La~Ni\~na phase to precede an El~Ni\~no phase; in either case, any phase manifestation of ENSO is predicted several months in advance, and its intensity is also typically predictable with the same anticipation.}. Markets then  anticipate climate disruptions which will influence crop yields, and commodities supply and demand, affecting futures and options markets' volatility. 

These external forces can have complex effects on market uncertainty. They affect more than the at-the-money implied-volatility level. They also shift the skew and curvature of the smile, which are the parts of the surface that matter most for out-of-the-money options, tail risk protection, and hedging books exposed to non-linear payoffs. A standard way to price commodity options is to begin with a tractable stochastic-volatility model and then calibrate it to a cross-section of option prices. This practice has clear advantages. The Black futures option formula remains a natural benchmark for commodity derivatives \citep{Black1976}. Affine or nearly affine stochastic-volatility models, such as Heston, Stein--Stein, and Sch\"obel--Zhu, provide richer dynamics while keeping pricing relatively manageable \citep{Heston1993, SteinStein1991, SchobelZhu1999}. More general stochastic-volatility specifications also have a long history in option pricing \citep{HullWhite1987}. 

The difficulty is that mathematical convenience can dominate the choice of model. In that case, the model may fit a given option chain on one day but still miss the time-series structure of the implied-volatility surface. This issue is especially important in commodity markets. Commodity prices and option prices are shaped by mean reversion, seasonal uncertainty, and shocks that are not always spanned by futures returns alone. The classical commodity-pricing literature has emphasized the stochastic behavior of spot and futures prices, the role of mean reversion, and the practical implications for valuation and hedging \citep{Schwartz1997, CortazarSchwartz2003}. For agricultural markets, seasonality and the Samuelson effect are central empirical features \citep{Samuelson1965, RichterSorensen2002, SchneiderTavin2018}. Multi-factor commodity models also show that volatility can contain risks that are not fully captured by the futures curve \citep{TrolleSchwartz2009}. These findings suggest that a useful option-pricing model for soybeans should not only match prices at one date. It should also respect the observed dynamics of the volatility surface over time.

\subsection{From fitting an option chain to reading the surface}

 Instead of selecting a stochastic-volatility model first, we begin with the observed implied-volatility surface. We use daily CME CVOL soybean indicators from October 2013 to August 2025 and study four quantities: the 30-day at-the-money level $L_t$, the additive skew $S_t$, the skew ratio $R_t$, and the convexity measure $C_t$. These indicators, which are all part of CME CVOL, summarize the center, asymmetry, and curvature of the 30-day soybean implied-volatility surface. They are attractive for statistical modeling because they are available as daily time series and are constructed from the option surface rather than from one isolated option quote. The empirical evidence points to a heterogeneous surface. The at-the-money (ATM) level is positive, right-skewed, and strongly mean-reverting. The convexity indicator is also positive and mean-reverting. The additive skew is different: it can be positive or negative, and negative values have an economic interpretation, as they correspond to periods when downside protection is more expensive than the benefit of upside exposure. This sign behavior immediately rules out positive-only diffusions for the additive skew. Hence, the data do not support a single diffusion class for all surface components. They call for a mixed system which still respects mean reversion: for instance logarithmic Ornstein-Uhlenbeck (log-OU) dynamics for positive surface variables and Ornstein-Uhlenbeck (OU) dynamics for the signed skew, provide among the simplest models that respect mean reversion and these sign properties.

This observation leads to the paper's main idea. The implied-volatility surface is not only an object to be fitted after a model has been chosen. It can also be used as an input for model selection, parameter identification, and model testing. In particular, under a log-OU stochastic-volatility model for the futures volatility, the three observable indicators $(L_t, S_t, C_t)$ contain direct information about the latent volatility state, the leverage-vol-of-vol product, and the smile-implied vol-of-vol. This turns the surface into an identification system.

\subsection{Relation of our paper to the literature}
 The literature on commodity futures and commodity option modeling provides the background for this paper. The work of \citet{Schwartz1997} and \citet{CortazarSchwartz2003} established the importance of mean reversion and stochastic factors in commodity markets. For agricultural derivatives, \citet{RichterSorensen2002} and \citet{SchneiderTavin2018} document seasonal volatility patterns and crop-cycle effects, while \citet{TrolleSchwartz2009} shows that commodity derivatives can require volatility factors not spanned by the futures curve. Our paper is in this vein, but with a different focus. We use the observed implied-volatility surface itself to guide the stochastic-volatility specification.
 Heston's CIR variance model has a closed-form characteristic function due to its affine structure \citep{Heston1993}. The Stein--Stein and Sch\"obel--Zhu models use OU-type volatility dynamics, and they lead to tractable pricing formulas under suitable specifications \citep{SteinStein1991, SchobelZhu1999}. These models are powerful benchmarks, but their tractability comes from strong structural choices. In the log-OU model studied here, the log-volatility is mean-reverting, so volatility is positive and log-normal in the stationary distribution. This is empirically natural for positive volatility indicators, but it breaks the affine structure. The pricing problem, therefore, has to be handled using numerical methods such as Feynman--Kac-based Monte Carlo and finite-difference methods.

Looking at the dynamics of the implied-volatility surface, \citet{DumasFlemingWhaley1998} show that static implied-volatility functions can be unstable when used dynamically. \citet{Cont2002} emphasizes that implied-volatility surfaces have their own dynamics that should be modeled directly. \citet{Fengler2009} studies arbitrage-free smoothing of implied-volatility surfaces, while \citet{Bergomi2005} develops a forward-looking view of smile dynamics. Our paper shares the view that the surface is dynamic. The difference is that we use surface indicators not only to describe smile dynamics but also to identify the latent stochastic-volatility parameters used in option pricing.

This approach is also connected to the broader idea that one constant volatility-of-volatility parameter may be too restrictive. Stochastic vol-of-vol extensions of Heston-type dynamics were introduced to improve the joint calibration of equity and volatility-index options, for example, by fitting multifactor stochastic-volatility specifications to the term structure of the index-option smirk. SABR-type models also make clear that smile curvature is controlled by volatility-of-volatility, while correlation controls skew \citep{HaganKumarLesniewskiWoodward2002}. In the present paper, we do not introduce an additional vol-of-vol state only for calibration flexibility. We choose instead to remain more parsimonious and closer to surface-driven calibration: the state is read from the surface: convexity identifies an effective surface-implied vol-of-vol, while skew identifies an effective leverage coefficient.

 Statistical estimation and specification testing for continuous-time models also have an impact on the literature. Exact or likelihood-based estimation of discretely sampled diffusions is a central tool in financial econometrics \citep{AitSahalia2002}. Joint estimation under physical and risk-neutral measures has been studied in option-pricing settings \citep{ChernovGhysels2000, Pan2002}. Specification testing for continuous-time models is also important because a model can fit one dimension of the data while failing another \citep{HongLi2005}. This paper contributes to that point by developing an over-identification diagnostic based on the CVOL surface: the time-series dynamics of the ATM level imply one value of vol-of-vol, while smile curvature implies another. Because the first quantity is estimated from time-series dynamics and the second is an effective option-implied quantity, the gap should be interpreted as a test of their equality under a common-parameter restriction, not as a pure test of physical-measure misspecification. Rejection may reflect time variation in vol-of-vol, a volatility risk premium, jump compensation, or approximation error. The diagnostic is constructed explicitly in Appendix~\ref{app:spec-test}.

\subsection{Main Contributions}
\label{subsec:main_contribution}

We develop a surface-driven stochastic-volatility framework for soybean futures options. The main idea is that the implied-volatility surface should not be treated only as an object to fit after a model is chosen. Instead, the surface itself can guide model selection, identify relevant state variables, and provide daily dynamic inputs for pricing and hedging. 

Our first contribution is empirical and statistical. Using daily CME CVOL soybean indicators from October 2013 to August 2025, we show that the implied-volatility surface is heterogeneous. The ATM level and convexity are positive, right-skewed, and well described by log-OU dynamics. The additive skew, however, changes sign and must be modeled on the real line. This leads to a joint mean-reverting surface-factor model for
\[
(\log \mathcal L_t,\mathcal S_t,\log \mathcal C_t),
\]
where level, skew, and convexity are allowed to have different supports, different mean-reversion speeds, and correlated shocks. 

Our second contribution is mathematical. We prove a surface-identification result linking the CVOL indicators to the parameters of a log-OU stochastic-volatility model. The ATM level identifies the latent volatility state, the convexity indicator identifies an effective volatility-of-volatility process, and the skew indicator identifies an effective leverage process. In this sense, the surface not only describes option prices: it transforms the benchmark constant-parameter model into a dynamic system with stochastic volatility, stochastic vol-of-vol, and stochastic leverage:
\[
\mathcal L_t \longrightarrow \sigma_t,
\qquad
\mathcal C_t \longrightarrow \xi_t^{\mathrm{surf}},
\qquad
\mathcal S_t \longrightarrow \rho_t^{\mathrm{surf}}.
\]

Our third contribution is a pricing framework. We derive the risk-neutral pricing equation for the log-OU model and show how the surface-implied quantities can be used as pricing inputs. Since the raw surface-implied vol-of-vol and leverage may contain risk premia, jump compensation, and short-lived market noise, we treat them as effective risk-neutral state variables and regularize them before pricing. Monte Carlo simulation is used as the main offline numerical benchmark, while finite differences provide an independent cross-validation of the pricing PDE. Black--76 is used only to report prices in implied-volatility units. Finally, the paper provides a model-diagnostic interpretation. The constant-vol-of-vol log-OU model remains a useful benchmark, but a single constant parameter does not jointly reconcile the time-series dynamics of the ATM level with the option-implied curvature state under the identifying restriction used here. Because the two quantities need not coincide across physical and risk-neutral measures, the gap can also reflect volatility risk premia, jump compensation, or approximation error. The result nevertheless shows that soybean option smiles contain information beyond the ATM volatility level and motivates the surface-driven specification proposed in the paper.

\subsection{Organization of the Paper}
\label{subsec:organization}

Section~\ref{sec:data_empirical} introduces the CME CVOL soybean dataset, defines the surface indicators, and documents the main empirical facts. This section shows why level, skew, and convexity must be modeled separately: the ATM level and convexity are positive and mean-reverting, while the additive skew can be negative and therefore requires a model on the real line. Section~\ref{sec:factor_identification} develops the joint surface-factor model for
\[
(\log \mathcal L_t,\mathcal S_t,\log \mathcal C_t).
\]
It derives the exact joint transition law, the stationary covariance structure, and the effect of mean reversion on long-run correlations. The section then proves the main surface-identification theorem. This theorem shows how the CVOL surface identifies the latent volatility state, a surface-implied volatility-of-volatility process, and a surface-implied leverage process. Section~\ref{sec:pricing_numerics_validation} turns the identification result into an option-pricing framework. We derive the risk-neutral pricing equation, state the Feynman--Kac representation, and explain how the surface-implied quantities enter the pricing model. Monte Carlo simulation is used as the main pricing benchmark, while the finite-difference method is used as an independent validation tool on a grid of strikes and maturities. The same section also describes how prices are translated into Black--76 implied-volatility units and how hedging diagnostics are constructed. Section~\ref{sec:conclusion} concludes. It summarizes the main empirical and mathematical findings, discusses the role of surface-implied vol-of-vol and leverage in commodity option pricing, and outlines natural extensions, including seasonal dynamics, jump risk, and a fully dynamic risk-neutral model for the surface-implied state variables. It also discusses the timeliness of this paper in the context of the current expectation that the planet is entering a strong El Ni\~no event.  The appendices collect the transition-density derivations, a strengthened asymptotic derivation of the skew and convexity identification relations, the specification test announced above, the Monte Carlo algorithm, and the finite-difference implementation details.

\section{Data, Surface Indicators, and Empirical Evidence}
\label{sec:data_empirical}

This section introduces the CME CVOL soybean dataset, defines the surface indicators used throughout the paper, and documents the main empirical facts that motivate the stochastic-volatility model. The goal is not only to describe the data, but also to show which features of the implied-volatility surface must be respected by any pricing and hedging model. The evidence below points to three facts that are central to the rest of the paper: the surface is strongly
mean-reverting, its components have different distributional supports, and the level, skew, and convexity encode different risk factors.
\subsection{The CME CVOL Dataset}
\label{subsec:cvol_data}

Our primary data source is the CME Group Volatility Index (CVOL) for soybean futures, which CME Group publishes daily after the CBOT trading session closes. The CVOL methodology decomposes the 30-day model-free variance swap rate into surface components: a level, a skew, and a convexity, following the variance-risk decomposition of \citet{CarrWu2009}. Since these components are constructed from a broad cross-section of out-of-the-money calls and puts at each date, rather than from a single option contract at a single strike, they are less sensitive to bid--ask noise that affects individual option quotes. This makes them a natural object for studying the low-frequency dynamics of the implied-volatility surface.

The sample runs from October~1, 2013 to August~26, 2025, covering
$n = 2998$ trading days. There are no missing observations. The largest gap between consecutive dates is four calendar days, corresponding to holiday weekends, and all raw CVOL fields are complete throughout. Table~\ref{tab:summary_stats}
reports descriptive statistics for the raw CVOL fields and for the four surface indicators defined in Subsection~\ref{subsec:indicators}. Two features of the raw data are especially important. First, the underlying
soybean futures price ranged from 802.5 to 1{,}769 cents per bushel over the sample. This factor-of-two movement covers the post-COVID agricultural rally,
the 2021 La~Ni\~na drought, and the 2022 Russia--Ukraine supply shock. Second,
the 30-day ATM implied volatility ranged from 9.70\% in November 2017 to
39.51\% on June~29, 2021. This factor-of-four range shows that agricultural
option markets can move sharply over short horizons. Any useful model of the
soybean implied-volatility surface must therefore allow for persistence,
mean reversion, and regime-sensitive variation in volatility.

\begin{table}[H]
\centering
\caption{Descriptive statistics for the CME CVOL soybean dataset,
October~1, 2013 to August~26, 2025 ($n = 2{,}998$ trading days).
Futures prices are in cents per bushel; all volatility figures are in annualized
percentage points. Skewness and excess kurtosis use standard moment definitions.
\% Neg.\ reports the share of observations with a negative value. All figures in
this table were independently recomputed from the raw CVOL feed and match to the
reported precision.}
\label{tab:summary_stats}
\setlength{\tabcolsep}{4.2pt}
\renewcommand{\arraystretch}{1.15}
\resizebox{\textwidth}{!}{%
\begin{tabular}{lrrrrrrrrrr}
\toprule
Variable & Mean & Std & Min & $Q_{25}$ & Median & $Q_{75}$ & Max
 & Skew & Ex.\ Kurt & \% Neg. \\
\midrule
Futures price (\textcent/bu)
 & 1117.1 & 226.5 & 802.5 & 938.3 & 1028.5 & 1316.1 & 1769.0
 & 0.74 & $-$0.65 & 0.0 \\[2pt]
Total variance ($\mathrm{Var}_{30}$)
 & 20.01 & 4.66 & 10.55 & 16.91 & 19.16 & 22.39 & 43.19
 & 0.90 & 1.21 & 0.0 \\
ATM vol\; $\mathcal{L}_t$
 & 18.60 & 4.41 & 9.70 & 15.71 & 17.86 & 21.08 & 39.51
 & 0.81 & 0.92 & 0.0 \\
Upside var.\; ($\mathrm{UpVar}_{30}$)
 & 20.94 & 5.23 & 10.65 & 17.46 & 19.86 & 23.49 & 45.66
 & 1.02 & 1.37 & 0.0 \\
Downside var.\; ($\mathrm{DnVar}_{30}$)
 & 18.98 & 4.21 & 10.33 & 16.33 & 18.34 & 21.30 & 43.76
 & 0.76 & 1.06 & 0.0 \\[2pt]
Additive skew\; $\mathcal{S}_t$
 & 1.958 & 2.149 & $-$3.868 & 0.509 & 1.440 & 3.095 & 12.812
 & 1.01 & 1.37 & \textbf{14.5} \\
Skew ratio\; $\mathcal{R}_t$
 & 1.102 & 0.103 & 0.852 & 1.029 & 1.084 & 1.166 & 1.565
 & 0.74 & 0.41 & 0.0 \\
Convexity\; $\mathcal{C}_t$
 & 1.077 & 0.031 & 0.997 & 1.055 & 1.071 & 1.096 & 1.236
 & 0.89 & 1.20 & 0.0 \\
\bottomrule
\end{tabular}}
\end{table}

\subsection{Surface Indicator Definitions}
\label{subsec:indicators}

The raw CVOL feed delivers daily fields including $\mathrm{Atm30}$,
$\mathrm{Var30}$, $\mathrm{UpVar30}$, $\mathrm{DnVar30}$,
$\mathrm{Skew1\_30}$, $\mathrm{Skew2\_30}$, and $\mathrm{Conv30}$.
We work with four derived surface indicators. The relationships below were
verified numerically; the maximum absolute residual against the raw fields does
not exceed $10^{-4}$, which is consistent with rounding in the published data.

\begin{definition}[CME CVOL surface indicators]
\label{def:indicators}
The four 30-day surface indicators are defined as
\begin{align}
 \mathcal{L}_t &:= \mathrm{Atm30}_t,
 \label{eq:level_def}\\
 \mathcal{S}_t &:= \mathrm{UpVar}_{30,t} - \mathrm{DnVar}_{30,t},
 \label{eq:skew_def}\\
 \mathcal{R}_t &:= \frac{\mathrm{UpVar}_{30,t}}{\mathrm{DnVar}_{30,t}},
 \label{eq:skewratio_def}\\
 \mathcal{C}_t &:= \frac{\mathrm{Var}_{30,t}}{\mathcal{L}_t}.
 \label{eq:conv_def}
\end{align}
\end{definition}

The four indicators have different financial meanings.

\medskip\noindent

\textbf{Level $\mathcal{L}_t$:}
The 30-day ATM implied volatility is the most familiar single-number summary of the option surface. It is always strictly positive. In our sample, it is right-skewed, with sample skewness 0.81 and excess kurtosis 0.92. Both the Jarque--Bera and Shapiro--Wilk tests reject normality at conventional levels, confirming that the level process has heavier right tails than a Gaussian. This is one reason why a log-OU specification is a natural candidate for the level.

\textbf{Additive skew $\mathcal{S}_t$:}
This indicator measures the difference between upside and downside variance in the same units as the ATM volatility. Positive values indicate stronger call-side variance demand, which often appears during pre-harvest drought risk periods when the right tail of the soybean price distribution becomes more important. Negative values indicate the opposite: the market prices more downside risk
than upside variance, often after a large confirmed harvest, when put protection becomes more valuable. In the sample, $\mathcal{S}_t$ is negative on \textbf{434 of the 2{,}998 trading
days}, or 14.5\% of the sample. The most extreme negative episode occurred in late September 2014, when the US harvest confirmed record bushel yields, and the nearby soybean futures contract fell from about 940 cents to below 830 cents in under two weeks. This sign behavior is directly informative for model selection:
positive-only diffusions are not compatible with $\mathcal{S}_t$.

\textbf{Skew ratio $\mathcal{R}_t$:}
The ratio $\mathrm{UpVar}/\mathrm{DnVar}$ measures the same asymmetry as $\mathcal{S}_t$, but in multiplicative form. It is always positive, ranging from 0.852 to 1.565. The sample correlation with $\mathcal{S}_t$ is 0.951, indicating that the two representations carry nearly the same empirical information. We keep $\mathcal{R}_t$ because its positivity makes it useful for comparing additive and
Multiplicative descriptions of skew.

\textbf{Convexity $\mathcal{C}_t$:}
The ratio $\mathrm{Var}_{30}/\mathcal{L}_t$ measures how much of the total 30-day variance comes from out-of-the-money options beyond the ATM contribution. When $\mathcal{C}_t = 1$, the smile is approximately flat. Values above one indicate meaningful smile curvature. In the sample, $\mathcal{C}_t$ ranges from 0.997 to 1.236, with a mean of 1.077 and a standard deviation of 0.031. Its distribution is tight but strongly non-normal, reflecting occasional bursts of curvature during stress periods. Because $\mathcal{C}_t$ occasionally dips slightly below one (its minimum is 0.997), the convexity-based vol-of-vol recovery in Section~\ref{sec:factor_identification} uses a positive-part correction; see Remark~\ref{rem:convexity_floor}.

\begin{figure}[H]
\centering
\includegraphics[width=0.9\textwidth]{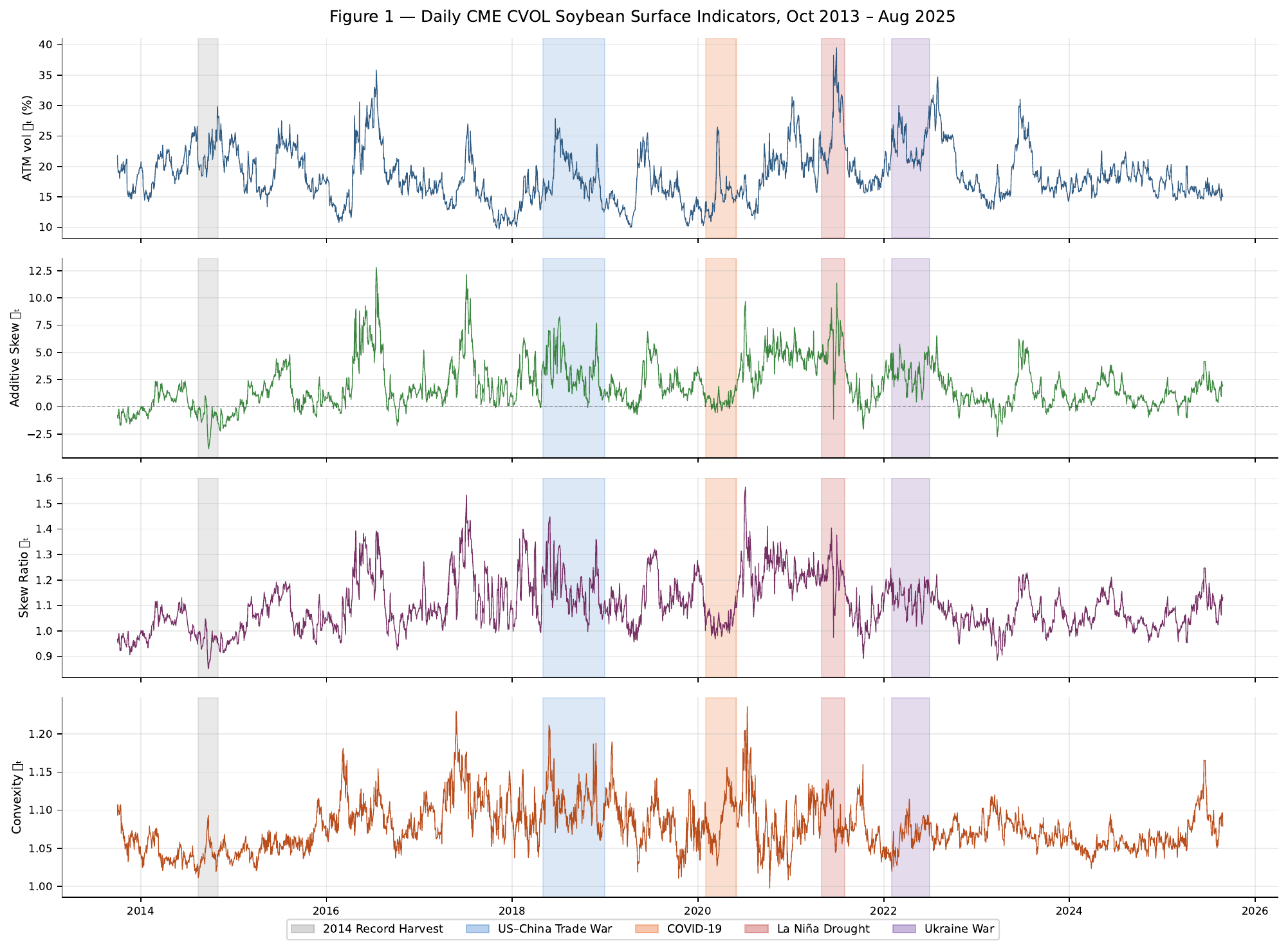}
\caption{Daily CME CVOL soybean surface indicators, October~2013 to
August~2025. Top: ATM implied volatility level $\mathcal{L}_t$ (\%).
Second: additive skew $\mathcal{S}_t$; the dashed line marks zero.
Third: skew ratio $\mathcal{R}_t$. Bottom: convexity $\mathcal{C}_t$.
Shaded regions mark key market episodes: the 2014 record US harvest,
the US-China trade war, COVID-19 supply disruption, the 2021 La~Ni\~na
drought, and the 2022 Russia-Ukraine supply shock.}
\label{fig:time_series}
\end{figure}
\subsection{Empirical Study}
\label{subsec:stylized_facts}

\subsubsection{Mean reversion and persistence}
\label{subsubsec:stationarity}

Figure~\ref{fig:time_series} plots the full daily time series of all four indicators from October 2013 to August 2025. The visual evidence already suggests mean reversion. All four series fluctuate around stable long-run levels instead of drifting without bound, and large deviations such as the level spike to 39.5\% in June 2021 or the negative-skew cluster in late 2014 resolve within weeks
or months.

To make the persistence structure precise, we apply two formal tests to each indicator. The augmented Dickey--Fuller (ADF) test takes the unit root as its null hypothesis, with the lag order selected by AIC. The KPSS test \citep{kwiatkowski1992testing} reverses the roles and takes stationarity as its null. Together, the two tests provide a useful bracketing assessment: a series that rejects the ADF null and fails to reject the KPSS null provides strong evidence of stationarity. Table~\ref{tab:stationarity} reports the results.

\begin{table}[H]
\centering
\caption{Stationarity tests for the four surface indicators ($n = 2998$).
ADF: augmented Dickey--Fuller test, $H_0$: unit root; lag order chosen by AIC
with constant included. KPSS: Kwiatkowski--Phillips--Schmidt--Shin test,
$H_0$: stationary, constant specification, automatic lag selection. Critical
values are for the 5\% level. Superscripts $^{*}$ and $^{**}$ denote rejection
at 5\% and 1\%, respectively.}
\label{tab:stationarity}
\renewcommand{\arraystretch}{1.2}
\begin{tabular}{lcccccc}
\toprule
& \multicolumn{3}{c}{ADF\quad ($H_0$: unit root)}
& \multicolumn{3}{c}{KPSS\quad ($H_0$: stationary)} \\
\cmidrule(lr){2-4}\cmidrule(lr){5-7}
Indicator & Stat. & Lags & $p$-value
 & Stat. & CV\;5\% & $p$-value \\
\midrule
Level\; $\mathcal{L}_t$
 & $-5.88$ & 5 & $<0.001^{**}$
 & $0.288$ & $0.463$ & $>0.10\phantom{^{*}}$ \\
Additive skew\; $\mathcal{S}_t$
 & $-5.40$ & 9 & $<0.001^{**}$
 & $0.637$ & $0.463$ & $0.019^{*}$ \\
Skew ratio\; $\mathcal{R}_t$
 & $-5.74$ & 9 & $<0.001^{**}$
 & $0.824$ & $0.463$ & $<0.010^{**}$ \\
Convexity\; $\mathcal{C}_t$
 & $-7.11$ & 5 & $<0.001^{**}$
 & $0.845$ & $0.463$ & $<0.010^{**}$ \\
\bottomrule
\end{tabular}
\end{table}

The ADF test rejects the unit-root null for all four indicators at the 0.1\% level, which supports mean reversion. The KPSS evidence is more nuanced. For the level $\mathcal{L}_t$, both tests agree: the ADF rejects a unit root and the KPSS does not reject stationarity. For the skew and convexity indicators, the KPSS statistic exceeds the 5\% critical value while the ADF still strongly rejects the unit-root null. This apparent conflict is common in the presence of near integration or structural breaks \citep{LeeHuangShin1997}. In this application, the sample spans several distinct market regimes, so slow changes in the unconditional mean can inflate the KPSS statistic without implying a unit root. The visual
evidence, the ADF results, and the half-life estimates below are therefore consistent with stationary but regime-sensitive mean reversion.
Table~\ref{tab:halflife} translates persistence into half-lives and annualized
mean-reversion speeds implied by the lag-1 autocorrelation.

\begin{table}[H]
\centering
\caption{Autocorrelation structure and estimated half-lives. $\rho(\ell)$ denotes
the sample lag-$\ell$ autocorrelation. Half-life is
$-\ln 2/\ln \rho(1)$ in trading days. The annualized mean-reversion proxy is
$\hat\kappa = -252\ln\rho(1)$.}
\label{tab:halflife}
\renewcommand{\arraystretch}{1.2}
\begin{tabular}{lcccccccr}
\toprule
Indicator & $\rho(1)$ & $\rho(5)$ & $\rho(10)$ & $\rho(22)$
 & $\rho(66)$ & $\rho(126)$ & Half-life (d) & $\hat\kappa$ \\
\midrule
Level\; $\mathcal{L}_t$
 & 0.972 & 0.876 & 0.780 & 0.583 & 0.153 & 0.025 & 24.0 & 7.27 \\
Additive skew\; $\mathcal{S}_t$
 & 0.955 & 0.811 & 0.731 & 0.558 & 0.151 & 0.055 & 15.1 & 11.56 \\
Skew ratio\; $\mathcal{R}_t$
 & 0.951 & 0.795 & 0.709 & 0.538 & 0.212 & 0.133 & 13.9 & 12.57 \\
Convexity\; $\mathcal{C}_t$
 & 0.907 & 0.763 & 0.651 & 0.482 & 0.321 & 0.209 & 7.1 & 24.52 \\
\bottomrule
\end{tabular}
\end{table}

The table reveals a clear hierarchy. The ATM level is the most persistent, with a half-life of roughly one trading month. The skew reverts faster, with a half-life of about three trading weeks. Convexity is the fastest-reverting component, with a half-life slightly above one week. Economically, this ordering is natural. Broad uncertainty about crop outcomes tends to resolve over weeks,
while far-wing curvature can adjust quickly to new weather reports, crop progress news, or sudden changes in demand for tail protection.
All four indicators retain economically meaningful autocorrelation at the one-quarter horizon. This rules out white-noise models and very fast mean reversion. At the six-month horizon, the level autocorrelation is close to zero, which is consistent with convergence to a stationary distribution within a crop
year. Convexity, however, still shows non-trivial persistence, suggesting that its dynamics may contain more than one timescale.

\subsubsection{Distributional properties and the sign of the skew}
\label{subsubsec:distribution}

Figure~\ref{fig:distributions} shows the empirical density of each indicator alongside a fitted Gaussian curve. All four distributions depart from normality. The level and skew are positively skewed with heavy right tails, consistent with infrequent but large upward moves. The skew ratio and convexity exhibit lighter tails, but they are still non-Gaussian.

\begin{figure}[H]
\centering
\includegraphics[width=\textwidth]{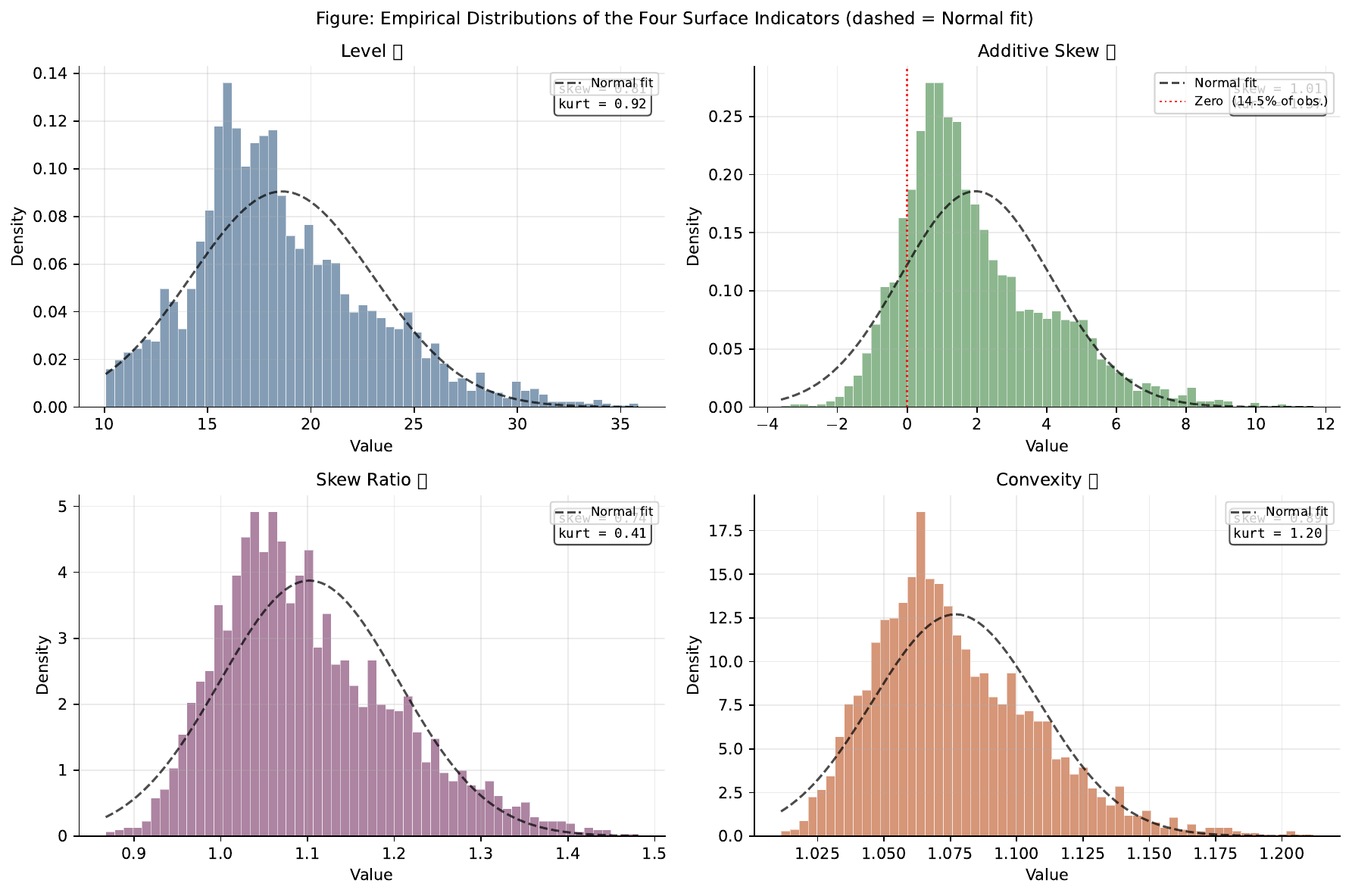}
\caption{Empirical density histograms for the four surface indicators, with a fitted Gaussian curve for comparison. The red dotted vertical line in the additive-skew panel marks zero, and the annotation reports the share of trading days for which $\mathcal{S}_t<0$. All panels use the October~2013 to August~2025
sample.}
\label{fig:distributions}
\end{figure}
The most important feature in Figure~\ref{fig:distributions} is the support of the additive skew. On 434 trading days, or 14.5\% of the sample, $\mathcal{S}_t$
is negative. Table~\ref{tab:neg_skew} reports the frequency of negative skew by calendar year.

\begin{table}[H]
\centering
\caption{Frequency of negative additive skew ($\mathcal{S}_t < 0$) by calendar year. The column ``\% of year'' reports the share of \emph{that year's own trading days present in the sample} on which $\mathcal{S}_t < 0$. Because the sample begins on October~1, 2013, the 2013 row covers only 64 trading days (Oct.--Dec.), not a full calendar year; all 64 of those days have negative skew, so the 2013 figure is 100.0\%, not a partial-year-weighted average. Likewise the 2025 row covers only 163 trading days (through Aug.~26).
}
\label{tab:neg_skew}
\renewcommand{\arraystretch}{1.1}
\begin{tabular}{lrr|lrr}
\toprule
Year & Days & \% of year & Year & Days & \% of year \\
\midrule
2013 & 64 & 100.0 & 2020 & 13 & 5.1 \\
2014 & 139 & 55.2 & 2021 & 23 & 9.1 \\
2015 & 45 & 17.7 & 2022 & 6 & 2.4 \\
2016 & 33 & 13.1 & 2023 & 50 & 20.0 \\
2017 & 2 & 0.8 & 2024 & 31 & 12.3 \\
2018 & 3 & 1.2 & 2025 & 10 & 6.1 \\
2019 & 15 & 6.0 & \textbf{Total} & \textbf{434} & \textbf{14.5} \\
\bottomrule
\end{tabular}
\end{table}
The global minimum is $\mathcal{S}_t=-3.868$, recorded on September~25, 2014. That week, USDA production estimates confirmed a historically large crop, and nearby soybean futures fell sharply. The option market repriced toward the downside
risk, so put-side variance became more expensive relative to call-side variance. This has a direct modeling implication. Any diffusion constrained to be strictly positive, such as CIR, log-OU, or geometric Brownian motion, is incompatible with the observed behavior of $\mathcal{S}_t$. The Ornstein--Uhlenbeck process, which
evolves on $\R$, is the only member of our candidate family that can naturally accommodate this sign-changing skew indicator. This is not a modeling preference; it is a constraint imposed by the data.

\subsubsection{Seasonality}
\label{subsubsec:seasonality}

Table~\ref{tab:seasonal} reports monthly averages of the level and additive skew, pooled across all years. Figure~\ref{fig:seasonality} gives the same information graphically.

\begin{table}[H]
\centering
\caption{Average ATM volatility level $\bar{\mathcal{L}}$ and additive skew $\bar{\mathcal{S}}$ by calendar month, pooled across October~2013 to August~2025. ``Std'' is the within-month standard deviation of the level.}
\label{tab:seasonal}
\renewcommand{\arraystretch}{1.2}
\begin{tabular}{lrrr|lrrr}
\toprule
Month & $\bar{\mathcal{L}}$ (\%) & Std & $\bar{\mathcal{S}}$
& Month & $\bar{\mathcal{L}}$ (\%) & Std & $\bar{\mathcal{S}}$ \\
\midrule
January & 17.4 & 4.4 & 1.45
& July & 22.7 & 5.5 & 4.41 \\
February & 16.5 & 3.8 & 1.52
& August & 19.8 & 4.0 & 1.82 \\
March & 18.1 & 4.1 & 1.59
& September& 19.1 & 2.9 & 1.39 \\
April & 16.9 & 3.5 & 1.48
& October & 18.2 & 3.3 & 0.83 \\
May & 18.2 & 3.1 & 2.52
& November & 17.2 & 3.6 & 1.19 \\
June & 21.9 & 5.2 & 4.02
& December & 17.2 & 3.4 & 1.19 \\
\bottomrule
\end{tabular}
\end{table}

\begin{figure}[H]
\centering
\includegraphics[width=\textwidth]{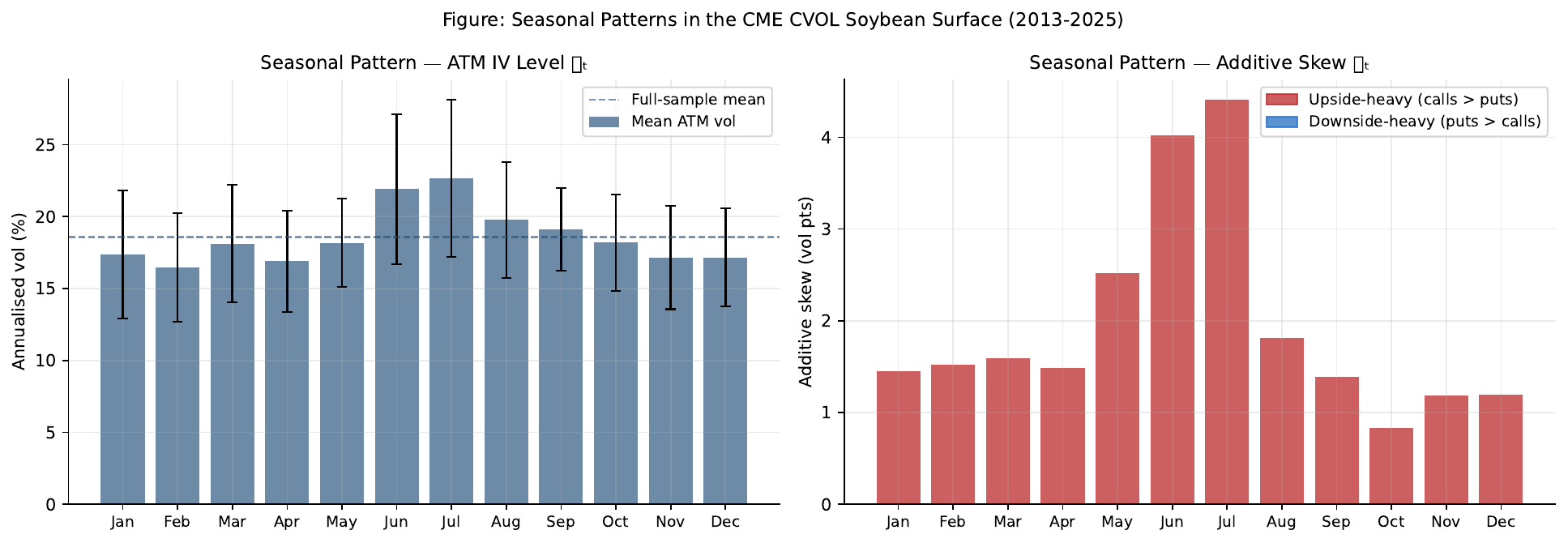}
\caption{Seasonal averages of the ATM volatility level $\mathcal{L}_t$ and the
additive skew $\mathcal{S}_t$, pooled across October~2013 to August~2025. The
pattern is consistent with the North American soybean crop calendar.}
\label{fig:seasonality}
\end{figure}

The seasonal pattern is pronounced. Volatility and skew peak in June and July, which correspond to the critical pollination window. During this period, temperature and precipitation uncertainty translate directly into price risk. The July ATM volatility average is 22.7\%, about 38\% above the February average of 16.5\%. The skew peaks even more strongly: the July average of 4.41 is more than five times the October average of 0.83. This reflects the concentration of upside demand for drought-risk calls during the growing season and its decline after the harvest outcome becomes clearer.

These seasonal patterns are well known in agricultural option markets. \citet{RichterSorensen2002} and \citet{SchneiderTavin2018}, for example, document seasonality in soybean and other agricultural futures and options. In the present paper, we keep the main model stationary and abstract from deterministic seasonality in the drift. This simplifies the identification and pricing analysis, but it also points to a natural extension discussed in the conclusion.

\subsubsection{Cross-indicator correlation structure}
\label{subsubsec:cross_corr}

Table~\ref{tab:crosscorr} and Figure~\ref{fig:corr_heatmap} report the pairwise
Pearson correlations among the four indicators.

\begin{table}[H]
\centering
\caption{Pairwise Pearson correlation matrix of the four surface indicators
($n=2{,}998$, October~2013 to August~2025).}
\label{tab:crosscorr}
\renewcommand{\arraystretch}{1.2}
\begin{tabular}{lcccc}
\toprule
& $\mathcal{L}$ & $\mathcal{S}$ & $\mathcal{R}$ & $\mathcal{C}$ \\
\midrule
Level\; $\mathcal{L}$
 & \phantom{$-$}1.000 & \phantom{$-$}0.437
 & \phantom{$-$}0.215 & $-0.230$ \\
Additive skew\; $\mathcal{S}$
 & \phantom{$-$}0.437 & \phantom{$-$}1.000
 & \phantom{$-$}0.951 & \phantom{$-$}0.354 \\
Skew ratio\; $\mathcal{R}$
 & \phantom{$-$}0.215 & \phantom{$-$}0.951
 & \phantom{$-$}1.000 & \phantom{$-$}0.427 \\
Convexity\; $\mathcal{C}$
 & $-0.230$ & \phantom{$-$}0.354
 & \phantom{$-$}0.427 & \phantom{$-$}1.000 \\
\bottomrule
\end{tabular}
\end{table}

\begin{figure}[H]
\centering
\includegraphics[width=0.5\textwidth]{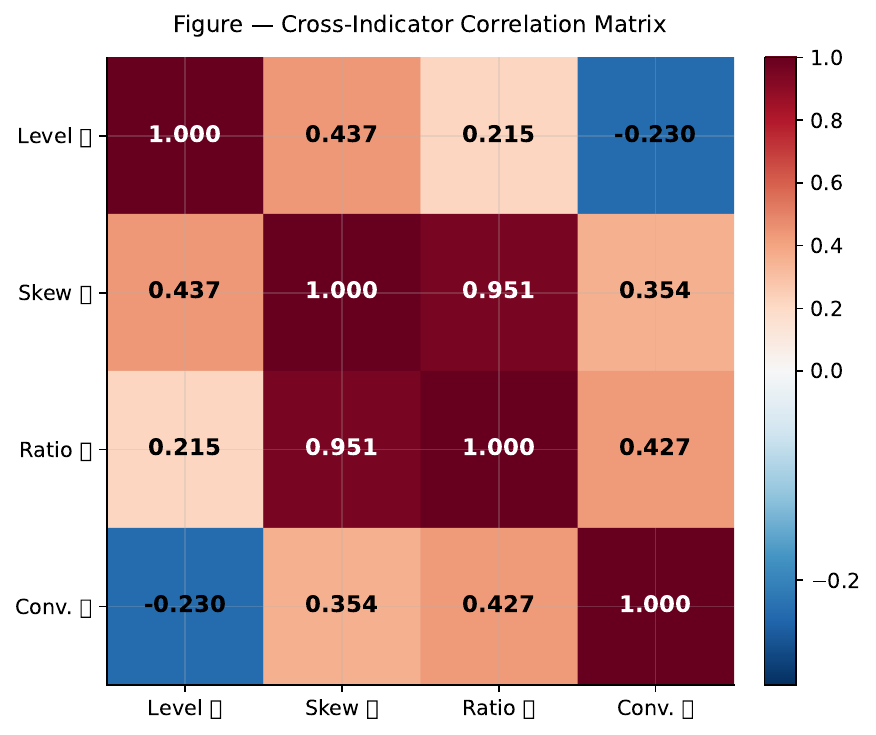}
\caption{Pairwise Pearson correlation heatmap for the four surface indicators.
The diagonal entries are one by construction.}
\label{fig:corr_heatmap}
\end{figure}

The correlation matrix contains three important messages. First, the correlation between $\mathcal{S}_t$ and $\mathcal{R}_t$ is 0.951, so the additive and multiplicative measures of skew are almost interchangeable empirically. Second, the correlation between the level and additive skew is positive, equal to 0.437.
This reflects a common agricultural-market pattern: high-volatility regimes such as drought scares, supply shocks, and trade-war uncertainty often come with stronger call-side demand. Third, the correlation between the level and convexity is negative, equal to $-0.230$. This is structurally important. When ATM volatility rises, the wings do not always reprice proportionally, so the smile can
become relatively flatter even though the whole surface is more expensive.

Figure~\ref{fig:rolling_corr} shows the rolling 252-day correlations between the
level and skew, and between the level and convexity.

\begin{figure}[H]
\centering
\includegraphics[width=\textwidth]{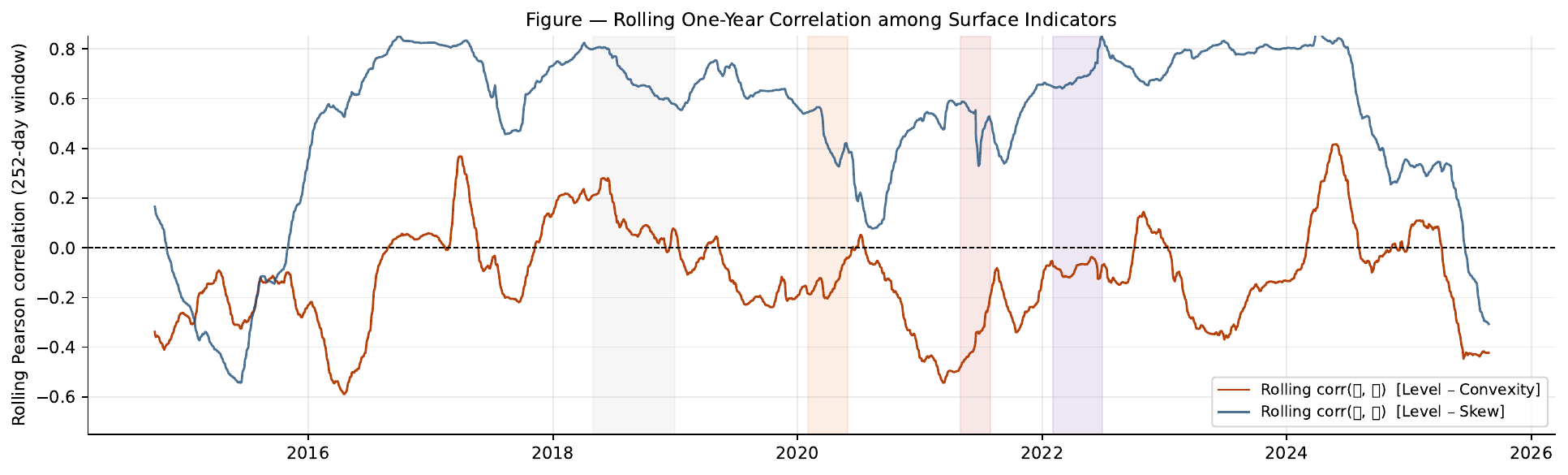}
\caption{Rolling 252-day Pearson correlations between
$(\mathcal{L}_t,\mathcal{C}_t)$ and $(\mathcal{L}_t,\mathcal{S}_t)$.
The dashed line marks zero. The level--convexity correlation is persistently negative across most of the sample, while the level--skew correlation is mostly positive but varies across market regimes.}
\label{fig:rolling_corr}
\end{figure}
The rolling correlations show that the level and curvature are driven by distinct risk factors. This is one motivation for treating the level, skew, and convexity as separate but correlated processes in the heterogeneous surface factor model developed in Section~\ref{sec:factor_identification}.

\subsection{Diffusion Model Selection for Surface Indicators}
\label{subsec:model_selection}

The central question of this section is a simple to state, but harder to answer: given a time series of daily surface
indicators, which mean-reverting diffusion fits best?
We consider three candidate processes, derive their exact transition
densities so that we can write down the likelihood without any
approximation, estimate each model by maximum likelihood, and then
compare fits using the Akaike and Bayesian information criteria.
The answer turns out to differ across indicators in a way that
has direct implications for model choice in the pricing sections
that follow.

\subsubsection{The Candidate Model Family}
\label{subsubsec:candidate_models}

We compare three parsimonious mean-reverting diffusion specifications. The OU and CIR models can be written as
\begin{equation}
 d X_t = \kappa(\theta - X_t)dt + \xi g(X_t) d W_t,
 \label{eq:general_sde}
\end{equation}
with $g(x)=1$ and $g(x)=\sqrt{x}$, respectively. The log-OU model is specified instead in logarithmic coordinates. Here $\kappa>0$ is the speed of mean reversion, $\theta$ is the long-run level in the relevant coordinate, and $\xi>0$ governs the noise intensity.

\begin{definition}[The Three Candidate Processes]
\label{def:candidates}
The three models correspond to three choices \begin{align}
dX_t &= \kappa(\theta-X_t)dt+\xi dW_t,
&& \text{OU},\\
dX_t &= \kappa(\theta-X_t)dt+\xi \sqrt{X_t}\,dW_t,
&& \text{CIR},\\
d\log X_t &= \kappa(\theta^{\log}-\log X_t)dt+\xi dW_t,
&& \text{log-OU}.
\end{align}
\end{definition}
The OU process adds Gaussian shocks to the mean-reverting drift.
Because the noise is additive and independent of the current level,
the process can wander into negative values, a feature that
most volatility models try to avoid, but which turns out to be
empirically necessary for the additive skew indicator $\mathcal{S}_t$.
The CIR process scales the noise by $\sqrt{X_t}$, which keeps
the process positive as long as the Feller condition
$2\kappa\theta \geq \xi^2$ holds; it arises naturally in
interest-rate theory \citep{CoxIngersollRoss1985} and has been used as
a variance process in the Heston model \citep{Heston1993}.
The log-OU process is specified directly by making $\ln X_t$ an OU process. Consequently, $X_t$ remains positive and has a state-dependent diffusion coefficient proportional to $X_t$ when written in levels, while its level drift is nonlinear because of the log transformation. This is a natural specification for a positive quantity whose fluctuations tend to scale with its level, as one might expect for implied volatility.

\subsubsection{Exact Transition Densities}
\label{subsubsec:transition_densities}

A useful feature of these three models is that each admits
an exact closed-form transition density, which means we can
write down the log-likelihood without any discretization
or approximation error.
With $n = 2{,}997$ daily transitions (from 2{,}998 observations)
and a daily time step of $\Delta t = 1/252$, this matters.
Full derivations of the three densities below are collected in
Appendix~\ref{app:transition-density-proofs}; here we record the
results needed for estimation.

\paragraph{Ornstein--Uhlenbeck.}
\label{subsubsec:ou_density}

The OU SDE has the explicit solution
\begin{equation}
 X_{t+h} = \theta + (X_t - \theta)\,e^{-\kappa h}
 + \xi\int_t^{t+h} e^{-\kappa(t+h-s)}d W_s,
 \label{eq:ou_solution}
\end{equation}
and since the stochastic integral is a Gaussian with zero mean,
the transition law is
\begin{equation}
 X_{t+h} \mid X_t \;\sim\;
 N\!\left(
 \mu_t^{OU},\; s^2_{OU}
 \right),
 \label{eq:ou_transition}
\end{equation}
where
\begin{equation}
 \mu_t^{OU} = \theta + (X_t - \theta)\,e^{-\kappa h},
 \qquad
 s^2_{OU} = \frac{\xi^2}{2\kappa}\bigl(1 - e^{-2\kappa h}\bigr).
 \label{eq:ou_moments}
\end{equation}
Both the conditional mean and the conditional variance are known
analytically.
The mean decays exponentially from $X_t$ toward $\theta$ at rate $\kappa$,
and the variance saturates at the stationary variance
$\xi^2/(2\kappa)$ as $h \to \infty$.
The exact log-likelihood for a sample $\{X_{t_0}, X_{t_1}, \ldots, X_{t_n}\}$
with $h = \Delta t$ is therefore
\begin{equation}
 \ell^{OU}(\kappa, \theta, \xi) =
 \sum_{i=0}^{n-1}
 \ln \phi\!\left(\frac{X_{t_{i+1}} - \mu_{t_i}^{OU}}{s_{OU}}\right)
 - \ln s_{OU},
 \label{eq:ou_loglik}
\end{equation}
where $\phi$ is the standard normal density.

The stationary distribution of the OU process is
\begin{equation}
 X_\infty \;\sim\; N\!\left(\theta,\;\frac{\xi^2}{2\kappa}\right),
 \label{eq:ou_stationary}
\end{equation}
which is Gaussian and, in particular, symmetric around $\theta$
and supported on all of $\R$.

\paragraph{Cox--Ingersoll--Ross.}
\label{subsubsec:cir_density}

For the CIR process, the transition density
is not Gaussian. Conditioned on $X_t$, the random variable $X_{t+h}$ follows a scaled non-central chi-squared distribution.
Specifically, let
\begin{equation}
 c = \frac{2\kappa}{\xi^2\,(1 - e^{-\kappa h})},
 \qquad
 q = \frac{2\kappa\theta}{\xi^2} - 1,
 \qquad
 u = c\,X_t\,e^{-\kappa h},
 \qquad
 v = c\,X_{t+h}.
 \label{eq:cir_uvq}
\end{equation}
Then the transition density is
\begin{equation}
 p_{CIR}(X_{t+h} \mid X_t) = c\,
 \left(\frac{v}{u}\right)^{q/2}
 e^{-(u+v)}\,
 I_q\!\left(2\sqrt{uv}\right),
 \label{eq:cir_density}
\end{equation}
where $I_q(\cdot)$ is the modified Bessel function of the first kind
of order $q$ (derivation in Appendix~\ref{app:transition-density-proofs}).
In log form:
\begin{equation}
 \ln p_{CIR} =
 \ln c + \tfrac{q}{2}\ln\!\tfrac{v}{u} - u - v +
 \ln I_q\!\left(2\sqrt{uv}\right).
 \label{eq:cir_loglik}
\end{equation}
In practice, $I_q(z)$ grows exponentially in $z$, which can
cause overflow for large arguments.
We use the numerically stable evaluation
$\ln I_q(z) = \ln \widetilde{I}_q(z) + z$, where
$\widetilde{I}_q(z) = e^{-z} I_q(z)$ is the
exponentially scaled Bessel function available in standard
scientific computing libraries.

The Feller condition $2\kappa\theta \geq \xi^2$ ensures that the
process never reaches zero; when it holds, the stationary
distribution is Gamma:
\begin{equation}
 X_\infty \;\sim\;
 \mathrm{Gamma}\!\left(
 \frac{2\kappa\theta}{\xi^2},\;
 \frac{\xi^2}{2\kappa}
 \right),
 \label{eq:cir_stationary}
\end{equation}
which is right-skewed and supported on $(0,\infty)$.

\paragraph{Log-OU.}
\label{subsubsec:logou_density}

The log-OU model is specified as $d\ln X_t = \kappa(\theta^{\log} - \ln X_t)dt + \xi\,d W_t$,
so $Y_t := \ln X_t$ is an OU process with long-run mean
$\theta^{\log}$ and diffusion coefficient $\xi$.
The transition density for $Y_{t+h} \mid Y_t$ is Gaussian
with mean and variance given by \eqref{eq:ou_moments} applied
to $Y_t$.
Using the change-of-variables formula for densities, the
transition density for $X_{t+h}$ in the original scale is
\begin{equation}
 p_{\log}(X_{t+h} \mid X_t) =
 \frac{1}{X_{t+h}}\,
 \phi\!\left(
 \frac{\ln X_{t+h} - \mu_t^{\log}}{s_{OU}}
 \right) \cdot \frac{1}{s_{OU}},
 \label{eq:logou_density}
\end{equation}
where $\mu_t^{\log} = \theta^{\log} + (\ln X_t - \theta^{\log})e^{-\kappa h}$
and $s_{OU}$ is as in \eqref{eq:ou_moments}.
The $1/X_{t+h}$ factor is the Jacobian of the log transformation.
The exact log-likelihood is therefore
\begin{equation}
 \ell^{\log}(\kappa, \theta^{\log}, \xi) =
 \sum_{i=0}^{n-1}
 \left[
 \ln\phi\!\left(\frac{\ln X_{t_{i+1}} - \mu_{t_i}^{\log}}{s_{OU}}\right)
 - \ln s_{OU} - \ln X_{t_{i+1}}
 \right].
 \label{eq:logou_loglik}
\end{equation}

The stationary distribution of $X_t$ under log-OU is
log-normal:
\begin{equation}
 \ln X_\infty \;\sim\;
 N\!\left(\theta^{\log},\;\frac{\xi^2}{2\kappa}\right),
 \label{eq:logou_stationary}
\end{equation}
which is positive, right-skewed, and has heavier right tails
than the Gaussian. This qualitative shape is much closer to what we observed in the empirical distributions of $\mathcal{L}_t$, $\mathcal{R}_t$,
and $\mathcal{C}_t$ in Section~\ref{subsubsec:distribution}.
Figure~\ref{fig:stat_dist} confirms this visually by overlaying the estimated stationary distribution from each model on the
empirical histogram of the level indicator $\mathcal{L}_t$.

\begin{figure}[H]
\centering
\includegraphics[width=\textwidth]{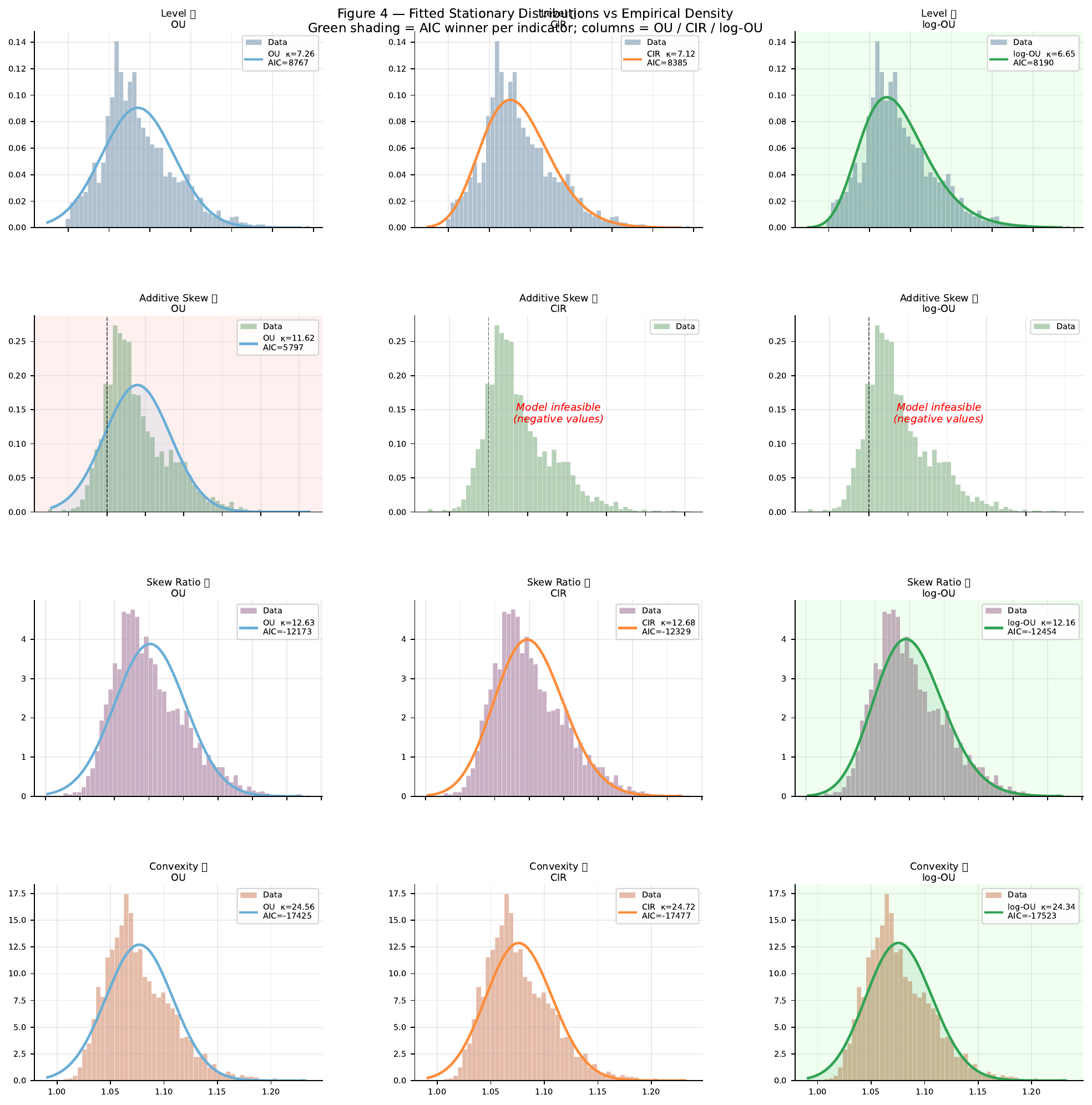}
\caption{Estimated stationary distributions for the three candidate
 models overlaid on the empirical density of the ATM vol level
 $\mathcal{L}_t$.
 \textit{Left}: OU (Gaussian, too symmetric and light-tailed).
 \textit{Center}: CIR (Gamma, right-skewed but lighter-tailed than the data).
 \textit{Right}: log-OU (log-normal, right-skewed, and the best fit by AIC).
 Parameters are the MLE estimates from Table~\ref{tab:mle_results}.
 Grey histogram: empirical density, $n = 2{,}998$.}
\label{fig:stat_dist}
\end{figure}

\begin{remark}[A note on non-negative volatility]
\label{rem:negativevol}
A common objection to the OU model for volatility is that
$\sigma_t$ can become negative. \citet{SchobelZhu1999} address this directly:
since $\sigma_t$ enters the futures dynamics as $\sigma_t F_t \dd W_t^{(1)}$, a negative value of $\sigma_t$ simply flips the sign of the Brownian driver and produces the same option prices as
$|\sigma_t|$. For the indicator $\mathcal{S}_t$, which is not itself a volatility but rather a signed difference of variance quantities, there is no such concern at all, as negative values are
economically meaningful, as Section~\ref{subsubsec:distribution}
established.
\end{remark}

\subsubsection{Maximum Likelihood Estimation}
\label{subsubsec:mle}

With exact transition densities, estimation is
straightforward: we maximise $\ell(\kappa, \theta, \xi)$ numerically
over the parameter vector for each indicator--model pair,
starting from multiple initial points to avoid local optima.
All three parameters are required to be positive; for the CIR
model, we additionally require the Feller condition $2\kappa\theta \geq \xi^2$. The time step is $\Delta t = 1/252$, matching the daily frequency of the data. Table~\ref{tab:mle_results} collects the MLE estimates, the exact log-likelihood values, and the AIC and BIC for every feasible indicator-model combination. A combination is marked ``n/a'' when the model is structurally infeasible: the log-OU and CIR models require strictly positive values, and since $\mathcal{S}_t$ takes negative values on 14.5\%
of sample days, neither can be applied to it.

\begin{table}[H]
\centering
\caption{MLE estimates and model fit statistics for each
 indicator--model pair. $n = 2{,}997$ transitions;
 $\Delta t = 1/252$. AIC $= -2\ell + 2k$ and
 BIC $= -2\ell + k\ln n$ with $k = 3$ parameters throughout.
 $\theta$ for log-OU is reported in the original scale
 (i.e., $e^{\hat\theta^{\log}}$).
 Bold entries indicate the best model for each indicator
 by AIC. Dashes (---) indicate structural infeasibility
 due to the sign constraint on $\mathcal{S}_t$.}
\label{tab:mle_results}
\renewcommand{\arraystretch}{1.25}
\setlength{\tabcolsep}{5pt}
\resizebox{\textwidth}{!}{%
\begin{tabular}{ll rr r rr r}
\toprule
Indicator & Model
 & $\hat\kappa$ & $\hat\theta$ & $\hat\xi$
 & $\ell$ & AIC & BIC \\
\midrule
\multirow{3}{*}{Level\; $\mathcal{L}_t$}
 & OU
 & 7.265 & 18.527 & 16.805
 & $-4380.4$ & $8766.7$ & $8784.7$ \\
 & CIR
 & 7.121 & 18.525 & 3.705
 & $-4189.5$ & $8385.0$ & $8403.0$ \\
 & \textbf{log-OU}
 & \textbf{6.651} & \textbf{18.027} & \textbf{0.842}
 & $\mathbf{-4091.9}$ & $\mathbf{8189.9}$ & $\mathbf{8207.9}$ \\
\midrule
\multirow{3}{*}{Additive skew\; $\mathcal{S}_t$}
 & \textbf{OU}
 & \textbf{11.624} & \textbf{1.979} & \textbf{10.328}
 & $\mathbf{-2895.7}$ & $\mathbf{5797.3}$ & $\mathbf{5815.3}$ \\
 & CIR & --- & --- & --- & --- & --- & --- \\
 & log-OU & --- & --- & --- & --- & --- & --- \\
\midrule
\multirow{3}{*}{Skew ratio\; $\mathcal{R}_t$}
 & OU
 & 12.635 & 1.103 & 0.516
 & $6089.7$ & $-12173.5$ & $-12155.5$ \\
 & CIR
 & 12.676 & 1.103 & 0.480
 & $6167.6$ & $-12329.1$ & $-12311.1$ \\
 & \textbf{log-OU}
 & \textbf{12.160} & \textbf{1.099} & \textbf{0.448}
 & $\mathbf{6229.8}$ & $\mathbf{-12453.7}$ & $\mathbf{-12435.7}$ \\
\midrule
\multirow{3}{*}{Convexity\; $\mathcal{C}_t$}
 & OU
 & 24.562 & 1.077 & 0.220
 & $8715.6$ & $-17425.1$ & $-17407.1$ \\
 & CIR
 & 24.724 & 1.077 & 0.210
 & $8741.7$ & $-17477.3$ & $-17459.3$ \\
 & \textbf{log-OU}
 & \textbf{24.337} & \textbf{1.076} & \textbf{0.201}
 & $\mathbf{8764.3}$ & $\mathbf{-17522.6}$ & $\mathbf{-17504.5}$ \\
\bottomrule
\end{tabular}}
\end{table}

\subsubsection{Model Comparison and the Heterogeneous Surface}
\label{subsubsec:model_comparison}

The pattern in Table~\ref{tab:mle_results} is unambiguous:
\begin{itemize}
 \item \textbf{Level $\mathcal{L}_t$}: log-OU wins by a wide margin.
 The AIC gap relative to OU is $8766.7 - 8189.9 = \mathbf{576.8}$
 points; relative to CIR it is $8385.0 - 8189.9 = \mathbf{195.1}$
 points. A gap of two AIC units is generally considered
 meaningful evidence; a gap of 195 is overwhelming.
 \item \textbf{Additive skew $\mathcal{S}_t$}: OU is the only
 feasible model. The data have negative values on 14.5\%
 of days, which rules out both CIR and log-OU before any likelihood
 is computed. OU wins not by being the best among equals but by
 being the only candidate that can fit the data at all.
 \item \textbf{Skew ratio $\mathcal{R}_t$}: log-OU wins over CIR
 by 124.6 AIC units and over OU by 280.2 units.
 \item \textbf{Convexity $\mathcal{C}_t$}: log-OU wins over CIR
 by 45.2 AIC units and over OU by 97.5 units.
\end{itemize}

The BIC, which penalizes model complexity more strongly for large
samples, yields identical winner rankings throughout,
since all three models have the same number of parameters
($k = 3$) and the likelihood differences dwarf any penalty
adjustment.
Figure~\ref{fig:aic_bic} visualizes these comparisons across
all four indicators.

\begin{figure}[H]
\centering
\includegraphics[width=\textwidth]{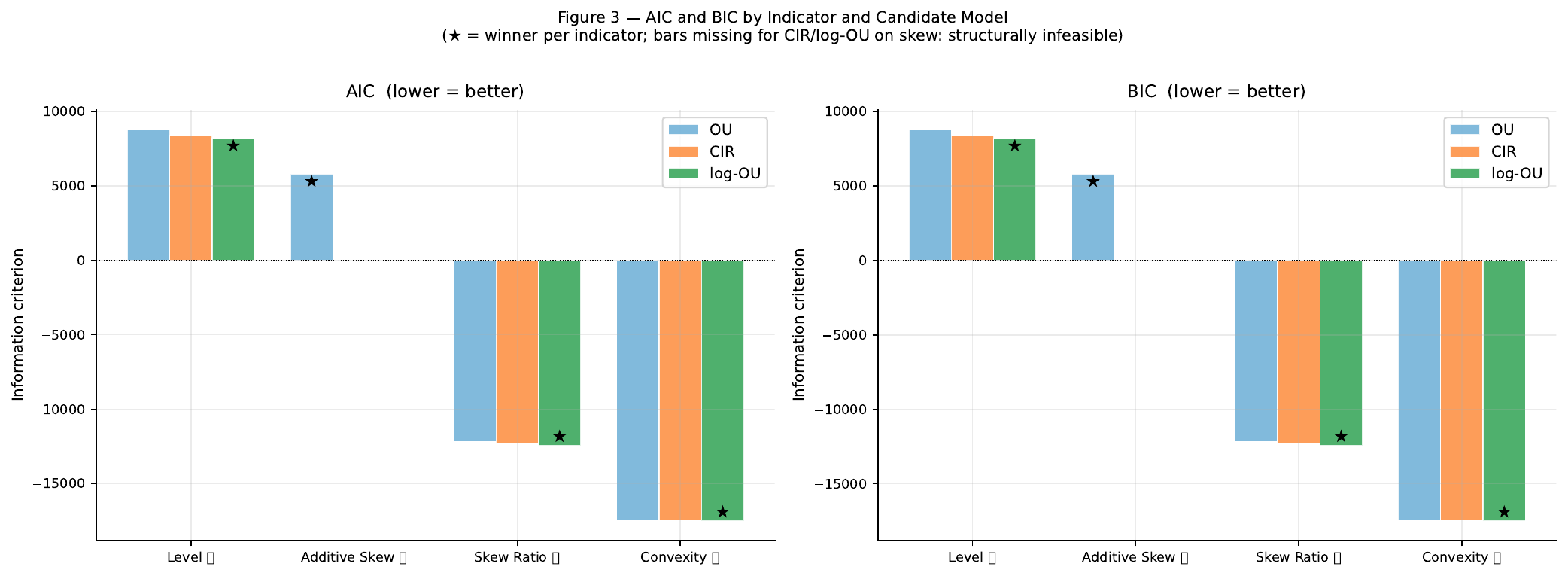}
\caption{AIC (left) and BIC (right) by indicator and candidate
 model. For the additive skew, only the OU bar is shown.
 CIR and log-OU are structurally infeasible.
 Note that for the skew ratio and convexity, the criterion
 values are negative (higher log-likelihoods than for
 the level and skew), so ``lower $=$ better'' means the
 log-OU bar is the most negative.}
\label{fig:aic_bic}
\end{figure}

\paragraph{Why Does log-OU Win for the Level?}

The magnitude of the AIC advantage for log-OU over OU
on the level indicator (576 units from 2997 observations)
is large enough to warrant a closer look at the mechanism. The stationary distribution of the OU process is Gaussian,
which is symmetric.
But $\mathcal{L}_t$ has a sample skewness of 0.81 and an excess
kurtosis of 0.92. It is right-skewed with a heavier right tail
than a Gaussian. A Gaussian stationary distribution cannot produce this shape without bias in the tail.
The log-OU process, by contrast, has a log-normal stationary
distribution, which is right-skewed by construction: for the estimated parameters, the log-normal places substantially
more probability mass in the 25--40\% vol region than the Gaussian
does.
Figure~\ref{fig:stat_dist} shows this graphically.
The log-normal curve follows the right shoulder of the empirical
histogram closely; the Gaussian underestimates the right tail
and overestimates the left shoulder.
The CIR's Gamma stationary distribution is also right-skewed,
which is why CIR beats OU, but it has lighter tails than the
log-normal, which is why log-OU still beats CIR. The same logic applies to the skew ratio and convexity, both of which share the right-skewed, always-positive character
of the level.

A few things are worth noting about the estimated parameters in
Table~\ref{tab:mle_results} before we proceed. The mean-reversion speeds $\hat\kappa$ are consistent across all
three models for a given indicator, which is reassuring: the
data contain a clear mean-reversion signal that all three
models pick up similarly.
For the level, $\hat\kappa$ is in the range 6.65--7.27 across
the three models; for the convexity, it is 24.3--24.7.
The ordering $\hat\kappa_C > \hat\kappa_S > \hat\kappa_L$ mirrors the half-life ordering from Table~\ref{tab:halflife} in
Section~\ref{sec:data_empirical}, as it should. The long-run level estimates $\hat\theta$ are also consistent
across models for each indicator and match the sample means
from Table~\ref{tab:summary_stats} closely: the level's
$\hat\theta \approx 18.0$--$18.5$ against a sample mean of
18.60, the skew ratio's $\hat\theta \approx 1.10$ against
a sample mean of 1.102.

The most striking contrast is in the volatility-of-volatility
estimates $\hat\xi$. For the level under OU, $\hat \xi = 16.81$, more than the entire range of the indicator. This is not a bug; it is a consequence of the additive noise structure. With additive Gaussian increments, the model needs a large
standard deviation to accommodate the wide range of the data.
Under log-OU, $\hat\xi = 0.842$, which is the volatility of
the log-vol process; that is, a more natural and interpretable
number (roughly an 84\% annual coefficient of variation
in log-volatility). This difference in scale is one reason the log-OU parameterization tends to be preferred in practice: the parameters are easier to reason about.

\subsubsection{The Heterogeneous Surface Factor: A Summary}
\label{subsubsec:heterogeneous_summary}

The findings of this section can be summarized in a single table and a single sentence.
Table~\ref{tab:model_winner} records the winning model for
each indicator along with the AIC margin of victory.
The sentence: \emph{no single diffusion class fits all features
of the soybean implied-volatility surface; the data
themselves --- not any modeling preference --- determine which
class is appropriate for each feature.}

\begin{table}[H]
\centering
\caption{Summary of model selection results.
 ``AIC margin'' is the difference between the second-best
 and best AIC; for the additive skew, there is only one
 feasible model.}
\label{tab:model_winner}
\renewcommand{\arraystretch}{1.3}
\begin{tabular}{llrrl}
\toprule
Indicator & Winner & $\hat\kappa$ & AIC margin & Reason \\
\midrule
Level\; $\mathcal{L}_t$
 & log-OU & 6.65 & 195.1 vs CIR & Right-skewed, always positive \\
Additive skew\; $\mathcal{S}_t$
 & OU & 11.62 & --- (only option) & Negative on 14.5\% of days \\
Skew ratio\; $\mathcal{R}_t$
 & log-OU & 12.16 & 124.6 vs CIR & Right-skewed, always positive \\
Convexity\; $\mathcal{C}_t$
 & log-OU & 24.34 & 45.2 vs CIR & Right-skewed, near 1 but always $>0$ \\
\bottomrule
\end{tabular}
\end{table}

This heterogeneity --- log-OU for level, convexity, and skew ratio; OU for the additive skew --- is the main empirical finding
of the paper. It is also the key input to the pricing model: because the level indicator $\mathcal{L}_t$ follows log-OU under the physical measure, the appropriate stochastic volatility specification for the futures option model is one in which $\ln\sigma_t$ follows an OU process, not one in which $\sigma_t$ or $\sigma_t^2$ does so.

\section{Heterogeneous Surface Factor Model and Surface Identification}
\label{sec:factor_identification}

The previous section shows that the implied-volatility surface is not well described by a single diffusion class. The level and convexity are positive, right-skewed, and best described on a log scale. The additive skew differs: it takes negative values on a meaningful part of the sample, so it requires a model on the full real line. This section builds a joint stochastic model that respects those empirical constraints. It then connects the observable surface indicators to the state variables and parameters of a log-OU stochastic-volatility model for soybean futures options.

We specify and estimate a multivariate surface factor model for the daily dynamics of the level, skew, and convexity. Then, mathematically, we prove that the three observable indicators identify the latent spot volatility, the leverage--vol-of-vol product, and the smile-implied vol-of-vol. This identification result bridges surface dynamics and option pricing.

\subsection{A Heterogeneous Surface Factor Model}
\label{subsec:heterogeneous_factor_model}

Let $
 Y_t = \bigl(Y_t^L,Y_t^S,Y_t^C\bigr)^{\top}
 := \bigl(\log \mathcal{L}_t,\mathcal{S}_t,\log \mathcal{C}_t\bigr)^{\top}.
$
We use logarithms for the level and convexity because these variables are positive and were selected as log-OU in Section~\ref{subsec:model_selection}. We keep the additive skew in its original scale because it can be negative and was selected as an OU process.

\begin{definition}[Heterogeneous surface factor model]
\label{def:heterogeneous_surface_model}
On a filtered probability space $(\Omega,\mathcal{F},(\mathcal{F}_t)_{t\geq 0},\mathbb{P})$, the three-dimensional surface factor model is
\begin{align}
 d\log \mathcal{L}_t
 &= \kappa_L\bigl(\theta_L-\log \mathcal{L}_t\bigr)dt
 + \xi_L\,dW_t^L,
 \label{eq:joint_level} \\
 d\mathcal{S}_t
 &= \kappa_S\bigl(\theta_S-\mathcal{S}_t\bigr)dt
 + \xi_S\,dW_t^S,
 \label{eq:joint_skew} \\
 d\log \mathcal{C}_t
 &= \kappa_C\bigl(\theta_C-\log \mathcal{C}_t\bigr)dt
 + \xi_C\,dW_t^C
 \label{eq:joint_convexity}
\end{align}
where $\kappa_j>0$ and $\xi_j>0$ for $j\in\{L,S,C\}$. The Brownian motion
$W_t=(W_t^L,W_t^S,W_t^C)^{\top}$ has instantaneous correlation matrix
\begin{equation}
 R_W=
 \begin{pmatrix}
 1 & \rho_{LS} & \rho_{LC} \\
 \rho_{LS} & 1 & \rho_{SC} \\
 \rho_{LC} & \rho_{SC} & 1
 \end{pmatrix},
 \qquad
 d\langle W^j,W^k\rangle_t=\rho_{jk}\,dt .
 \label{eq:brownian_corr}
\end{equation}
The full parameter vector is
\[
 \Theta=\bigl(\kappa_L,\theta_L,\xi_L,
 \kappa_S,\theta_S,\xi_S,
 \kappa_C,\theta_C,\xi_C,
 \rho_{LS},\rho_{LC},\rho_{SC}\bigr).
\]
\end{definition}

This specification is deliberately heterogeneous. It does not force the same support condition on all parts of the surface. The variables $\mathcal{L}_t$ and $\mathcal{C}_t$ remain positive by construction, while $\mathcal{S}_t$ is allowed to move on $\R$. The correlations in \eqref{eq:brownian_corr} capture contemporaneous co-movement among shocks to the surface. They do not impose cross-drift effects. After a shock, each component mean-reverts at its own speed.

\subsection{Exact Joint Transition Density}
\label{subsec:joint_transition_density}

The model in Definition~\ref{def:heterogeneous_surface_model} is Gaussian in the transformed state vector $Y_t$. For a fixed time step $h>0$, each component has the exact solution
\begin{equation}
 Y_{t+h}^j
 = \theta_j + \bigl(Y_t^j-\theta_j\bigr)e^{-\kappa_j h}
 + \xi_j \int_t^{t+h} e^{-\kappa_j(t+h-u)}\,dW_u^j,
 \qquad j\in\{L,S,C\}.
 \label{eq:joint_exact_solution}
\end{equation}
Therefore,
\begin{equation}
 Y_{t+h}\mid Y_t \sim
 \mathcal{N}\bigl(m_h(Y_t),\Sigma_h\bigr),
 \label{eq:joint_transition}
\end{equation}
where
\begin{equation}
 m_h^j(Y_t)=\theta_j+\bigl(Y_t^j-\theta_j\bigr)e^{-\kappa_j h},
 \label{eq:joint_mean}
\end{equation}
and the $(j,k)$ entry of the transition covariance matrix is
\begin{equation}
 [\Sigma_h]_{jk}
 = \rho_{jk}\xi_j\xi_k
 \frac{1-e^{-(\kappa_j+\kappa_k)h}}{\kappa_j+\kappa_k},
 \qquad \rho_{jj}=1.
 \label{eq:joint_covariance}
\end{equation}

\begin{proposition}[Joint transition covariance]
\label{prop:joint_covariance}
Under Definition~\ref{def:heterogeneous_surface_model}, the conditional covariance of $Y_{t+h}$ given $Y_t$ is given by \eqref{eq:joint_covariance}. In particular,
\[
 [\Sigma_h]_{jj}
 = \xi_j^2\frac{1-e^{-2\kappa_jh}}{2\kappa_j},
\]
which is the usual exact transition variance of a scalar OU process.
\end{proposition}

\begin{proof}
From \eqref{eq:joint_exact_solution}, the only random part of $Y_{t+h}^j$ is
\[
 \xi_j\int_t^{t+h}e^{-\kappa_j(t+h-u)}\,dW_u^j.
\]
For two components $j$ and $k$, It\^o isometry for correlated Brownian motions gives
\begin{align*}
 \operatorname{Cov}(Y_{t+h}^j,Y_{t+h}^k\mid Y_t)
 &= \xi_j\xi_k
 \int_t^{t+h} e^{-\kappa_j(t+h-u)}e^{-\kappa_k(t+h-u)}
 d\langle W^j,W^k\rangle_u \\
 &= \rho_{jk}\xi_j\xi_k
 \int_t^{t+h} e^{-(\kappa_j+\kappa_k)(t+h-u)}\,du \\
 &= \rho_{jk}\xi_j\xi_k
 \frac{1-e^{-(\kappa_j+\kappa_k)h}}{\kappa_j+\kappa_k}.
\end{align*}
Setting $j=k$ and $\rho_{jj}=1$ gives the diagonal formula.
\end{proof}

Formula \eqref{eq:joint_covariance} is important because it shows how mean reversion dampens long-run dependence. The covariance between two surface factors is not controlled by $\rho_{jk}$ alone. It is also scaled by $\kappa_j+\kappa_k$. When one factor mean-reverts quickly, its long-run co-movement with the other factor is reduced.

\subsection{Likelihood and Two-Step Estimation}
\label{subsec:joint_likelihood}

Let $Y_{t_0},\ldots,Y_{t_n}$ be the daily transformed observations, with $h=1/252$. Define the one-step innovation
\begin{equation}
 \varepsilon_i(\Theta)=Y_{t_{i+1}}-m_h(Y_{t_i}),
 \qquad i=0,\ldots,n-1.
 \label{eq:innovation_vector}
\end{equation}
The Gaussian log-likelihood in the transformed coordinates is
\begin{equation}
 \ell_Y(\Theta)
 =-\frac{n}{2}\log |\Sigma_h|
 -\frac{1}{2}\sum_{i=0}^{n-1}
 \varepsilon_i(\Theta)^{\top}\Sigma_h^{-1}\varepsilon_i(\Theta)
 +\text{constant}.
 \label{eq:joint_loglik_y}
\end{equation}
If the likelihood is written for the original variables $(\mathcal{L}_t,\mathcal{S}_t,\mathcal{C}_t)$ rather than for $Y_t$, the Jacobian contribution is
\begin{equation}
 -\sum_{i=0}^{n-1}\bigl(\log \mathcal{L}_{t_{i+1}}+\log \mathcal{C}_{t_{i+1}}\bigr),
 \label{eq:jacobian_joint}
\end{equation}
which comes from the transformations $\log\mathcal{L}$ and $\log\mathcal{C}$.

In practice, we estimate the model in two steps. First, the marginal parameters $(\kappa_j,\theta_j)$ are estimated from the exact one-dimensional likelihoods in Section~\ref{subsec:model_selection}. Second, keeping these estimates fixed, we estimate $(\xi_L,\xi_S,\xi_C,\rho_{LS},\rho_{LC},\rho_{SC})$ from the joint likelihood \eqref{eq:joint_loglik_y}. This second step is a standard multivariate Gaussian covariance estimation problem with a structured covariance matrix.

\begin{proposition}[Identification of the joint covariance parameters]
\label{prop:joint_identification}
Suppose $h>0$ is fixed, $\kappa_j>0$ is known for each $j$, and $R_W$ is positive definite. Then the map
\[
 (\xi_L,\xi_S,\xi_C,\rho_{LS},\rho_{LC},\rho_{SC})
 \longmapsto \Sigma_h
\]
defined by \eqref{eq:joint_covariance} is one-to-one on the parameter space $\xi_j>0$ and $|\rho_{jk}|<1$ subject to positive definiteness.
\end{proposition}

\begin{proof}
The diagonal entries of $\Sigma_h$ identify $\xi_j$ uniquely because
\[
 [\Sigma_h]_{jj}
 = \xi_j^2\frac{1-e^{-2\kappa_jh}}{2\kappa_j},
\]
and the factor $(1-e^{-2\kappa_jh})/(2\kappa_j)$ is strictly positive. Hence
\[
 \xi_j
 =\left(
 [\Sigma_h]_{jj}\frac{2\kappa_j}{1-e^{-2\kappa_jh}}
 \right)^{1/2}.
\]
Once $\xi_j$ and $\xi_k$ are known, the off-diagonal entry identifies $\rho_{jk}$ through
\[
 \rho_{jk}
 = [\Sigma_h]_{jk}
 \frac{\kappa_j+\kappa_k}{\xi_j\xi_k\{1-e^{-(\kappa_j+\kappa_k)h}\}}.
\]
Thus no two distinct parameter vectors can produce the same transition covariance matrix.
\end{proof}

\subsection{Stationary Covariance and Correlation}
\label{subsec:stationary_covariance}

Letting $h\to\infty$ in the covariance calculation gives the stationary covariance matrix of the transformed surface factors:
\begin{equation}
 [\Sigma_\infty]_{jk}
 = \rho_{jk}\frac{\xi_j\xi_k}{\kappa_j+\kappa_k},
 \qquad
 [\Sigma_\infty]_{jj}=\frac{\xi_j^2}{2\kappa_j}.
 \label{eq:stationary_covariance}
\end{equation}
Therefore, the stationary correlation between two factors is
\begin{equation}
 \operatorname{Corr}_{\infty}(Y^j,Y^k)
 = 2\rho_{jk}\frac{\sqrt{\kappa_j\kappa_k}}{\kappa_j+\kappa_k}.
 \label{eq:stationary_correlation}
\end{equation}

\begin{proposition}[Mean-reversion dampening of long-run correlation]
\label{prop:correlation_dampening}
For any pair $j\neq k$, the stationary correlation satisfies
\begin{equation}
 \left|\operatorname{Corr}_{\infty}(Y^j,Y^k)\right|
 \leq |\rho_{jk}|,
 \label{eq:corr_bound}
\end{equation}
with equality if and only if $\kappa_j=\kappa_k$.
\end{proposition}

\begin{proof}
From \eqref{eq:stationary_correlation}, the ratio between the stationary correlation and the Brownian correlation is
\[
 D_{jk}=2\frac{\sqrt{\kappa_j\kappa_k}}{\kappa_j+\kappa_k}.
\]
By the arithmetic--geometric mean inequality,
\[
 2\sqrt{\kappa_j\kappa_k}\leq \kappa_j+\kappa_k,
\]
so $0<D_{jk}\leq 1$. Equality holds exactly when $\kappa_j=\kappa_k$. This proves the result.
\end{proof}

This result has a simple interpretation. Two factors may be strongly correlated at the shock level, but if one factor mean-reverts much faster than the other, their long-run level correlation is weaker. This is relevant for the soybean surface because convexity mean-reverts faster than the level. As a result, the long-run level--convexity correlation is smaller in magnitude than the instantaneous Brownian correlation.

\subsection{Surface Identification of the Stochastic-Volatility Model}
\label{subsec:surface_identification}

We now connect the surface factors to the risk-neutral stochastic-volatility model used for
pricing. Under the risk-neutral measure \(\mathbb Q\), let the soybean futures price satisfy
\begin{align}
\frac{dF_t}{F_t}
&=
\sigma_t\,dB_t^{(1)},
\label{eq:rn_futures}\\
dX_t
&=
\kappa(\theta-X_t)dt+\xi\,dB_t^{(2)},
\qquad X_t=\log\sigma_t,
\label{eq:rn_logvol}\\
d\langle B^{(1)},B^{(2)}\rangle_t
&=
\rho\,dt.
\label{eq:rn_corr}
\end{align}
Here \(F_t\) is the futures price, \(\sigma_t\) is the instantaneous volatility, \(X_t\) is
log-volatility, \(\kappa>0\) is the mean-reversion speed, \(\theta\) is the long-run mean of
log-volatility, \(\xi>0\) is the volatility of log-volatility, and \(\rho\in[-1,1]\) is the
futures--volatility correlation.

For a fixed option maturity \(\tau>0\), the observed CVOL surface is summarized by three
quantities:
\[
\mathcal L_t=\text{ATM implied-volatility level},\qquad
\mathcal S_t=\text{additive skew},\qquad
\mathcal C_t=\text{convexity ratio}.
\]
The next theorem gives a leading-order map from these observable quantities to the main
risk-neutral volatility states. Step~1 of the proof (the level relation) is an exact
consequence of the OU transition law of $X_t$; Steps~2 and~3 (the skew and convexity
relations) are leading-order results in the small-vol-of-vol regime. A self-contained
derivation of Steps~2--3 via a cumulant expansion of the log-return, together with pointers
to the asymptotic-expansion literature that makes the argument fully rigorous, is given in
Appendix~\ref{app:identification-proof}.

\begin{theorem}[Leading-order surface identification system]
\label{thm:surface_identification}
Assume that the soybean futures price follows the log-OU stochastic-volatility model
\eqref{eq:rn_futures}--\eqref{eq:rn_corr}. Define
\begin{equation}
A(\kappa,\tau)
=
\frac{1-e^{-\kappa\tau}}{\kappa\tau},
\label{eq:A_factor}
\end{equation}
\begin{equation}
G(\kappa,\tau)
=
\frac{1}{2(\kappa \tau)^2}
\left\{
\tau
-
\frac{2}{\kappa}\left(1-e^{-\kappa \tau}\right)
+
\frac{1}{2\kappa}\left(1-e^{-2\kappa \tau}\right)
\right\},
\label{eq:G_factor}
\end{equation}
\begin{equation}
J(\kappa,\xi,\tau)
=
\xi^2G(\kappa,\tau),
\label{eq:J_factor}
\end{equation}
and
\begin{equation}
f(\kappa,\tau)
=
\frac{2}{\tau}
\left\{
\frac{\tau}{\kappa}
-
\frac{1-e^{-\kappa\tau}}{\kappa^2}
\right\}.
\label{eq:f_factor}
\end{equation}
Then, under the usual small-maturity and small-smile approximation, the CVOL level, skew,
and convexity satisfy
\begin{equation}
\log\mathcal L_t
\approx
\theta
+
A(\kappa,\tau)(\log\sigma_t-\theta)
+
J(\kappa,\xi,\tau),
\label{eq:level_identification}
\end{equation}
\begin{equation}
\mathcal S_t
\approx
\rho\xi\sigma_t f(\kappa,\tau),
\label{eq:skew_identification}
\end{equation}
and
\begin{equation}
\mathcal C_t-1
\approx
\xi^2G(\kappa,\tau).
\label{eq:convexity_identification}
\end{equation}
Consequently, if \(\kappa\) and \(\theta\) are known or consistently estimated, the surface gives
the leading-order recovery formulas
\begin{equation}
\sigma_t
\approx
\exp\left\{
\theta+
\frac{\log\mathcal L_t-\theta-J(\kappa,\xi,\tau)}
{A(\kappa,\tau)}
\right\},
\label{eq:sigma_identified}
\end{equation}
\begin{equation}
\rho\xi
\approx
\frac{\mathcal S_t}{\sigma_t f(\kappa,\tau)},
\label{eq:rho_xi_identified}
\end{equation}
and
\begin{equation}
\xi^2
\approx
\frac{\mathcal C_t-1}{G(\kappa,\tau)}.
\label{eq:xi_identified}
\end{equation}
\end{theorem}

\begin{proof}
The proof separates the three pieces of the surface. The ATM level is linked to the maturity
average of expected future log-volatility. The skew is linked to the first-order interaction between
futures shocks and volatility shocks. The convexity is linked to the second-order uncertainty of
future volatility.

\medskip
\noindent\textbf{Step 1: The ATM level and the maturity-averaged log-volatility.}

Let \(X_t=\log\sigma_t\). From \eqref{eq:rn_logvol}, \(X_t\) follows an Ornstein--Uhlenbeck
process. Its exact solution over the interval \([t,t+u]\) is
\begin{equation}
X_{t+u}
=
\theta+(X_t-\theta)e^{-\kappa u}
+
\xi\int_t^{t+u}e^{-\kappa(t+u-s)}\,dB_s^{(2)}.
\label{eq:ou_solution_identification}
\end{equation}
Taking conditional expectation at time \(t\) gives
\begin{equation}
\mathbb E_t[X_{t+u}]
=
\theta+(X_t-\theta)e^{-\kappa u}.
\label{eq:ou_conditional_mean_identification}
\end{equation}
An implied volatility with maturity \(\tau\) is not an instantaneous object. It reflects the
market price of volatility over the life of the option. For this reason, the relevant log-volatility
quantity is the maturity average
\begin{equation}
\bar X_{t,\tau}
=
\frac{1}{\tau}\int_0^\tau X_{t+u}\,du.
\label{eq:average_logvol_definition}
\end{equation}
Using \eqref{eq:ou_conditional_mean_identification}, we obtain
\begin{align}
\mathbb E_t[\bar X_{t,\tau}]
&=
\frac{1}{\tau}\int_0^\tau
\left\{
\theta+(X_t-\theta)e^{-\kappa u}
\right\}du \nonumber\\
&=
\theta
+
(X_t-\theta)\frac{1}{\tau}\int_0^\tau e^{-\kappa u}\,du \nonumber\\
&=
\theta
+
A(\kappa,\tau)(X_t-\theta),
\label{eq:average_logvol_mean_identification}
\end{align}
where
\[
A(\kappa,\tau)=\frac{1-e^{-\kappa\tau}}{\kappa\tau}.
\]
This part is exact and follows directly from the OU transition.

We now relate this maturity-averaged log-volatility to the ATM implied-volatility level. The
ATM implied volatility is obtained by inverting an option price, so it is not exactly equal to
\(\exp(\mathbb E_t[\bar X_{t,\tau}])\). However, for short maturities and moderate smiles, the
ATM implied volatility is well approximated by the volatility generated by the average future
log-volatility. Since \(\bar X_{t,\tau}\) is Gaussian conditional on \(\mathcal F_t\), the exponential
mapping introduces a Jensen correction:
\begin{equation}
\log\mathcal L_t
\approx
\mathbb E_t[\bar X_{t,\tau}]
+
\frac12\operatorname{Var}_t(\bar X_{t,\tau}).
\label{eq:atm_jensen_approximation}
\end{equation}

It remains to compute \(\operatorname{Var}_t(\bar X_{t,\tau})\). From
\eqref{eq:ou_solution_identification},
\[
\bar X_{t,\tau}
=
\mathbb E_t[\bar X_{t,\tau}]
+
\frac{\xi}{\tau}
\int_0^\tau
\int_t^{t+u}
e^{-\kappa(t+u-s)}\,dB_s^{(2)}du.
\]
Changing the order of integration gives
\[
\bar X_{t,\tau}
=
\mathbb E_t[\bar X_{t,\tau}]
+
\frac{\xi}{\tau}
\int_t^{t+\tau}
\left[
\int_s^{t+\tau}e^{-\kappa(v-s)}\,dv
\right]dB_s^{(2)}.
\]
The inner integral is
\[
\int_s^{t+\tau}e^{-\kappa(v-s)}\,dv
=
\frac{1-e^{-\kappa(t+\tau-s)}}{\kappa}.
\]
By It\^o isometry,
\begin{align}
\operatorname{Var}_t(\bar X_{t,\tau})
&=
\frac{\xi^2}{\tau^2\kappa^2}
\int_t^{t+\tau}
\left(1-e^{-\kappa(t+\tau-s)}\right)^2ds \nonumber\\
&=
\frac{\xi^2}{\tau^2\kappa^2}
\int_0^\tau
\left(1-e^{-\kappa z}\right)^2dz.
\label{eq:average_logvol_variance_step}
\end{align}
Expanding the square,
\[
\int_0^\tau
\left(1-e^{-\kappa z}\right)^2dz
=
\int_0^\tau
\left(1-2e^{-\kappa z}+e^{-2\kappa z}\right)dz,
\]
and therefore
\[
\int_0^\tau
\left(1-e^{-\kappa z}\right)^2dz
=
\tau
-
\frac{2}{\kappa}\left(1-e^{-\kappa\tau}\right)
+
\frac{1}{2\kappa}\left(1-e^{-2\kappa\tau}\right).
\]
Hence
\begin{equation}
\frac12\operatorname{Var}_t(\bar X_{t,\tau})
=
\frac{\xi^2}{2(\kappa\tau)^2}
\left\{
\tau
-
\frac{2}{\kappa}\left(1-e^{-\kappa\tau}\right)
+
\frac{1}{2\kappa}\left(1-e^{-2\kappa\tau}\right)
\right\}.
\label{eq:jensen_term_identification}
\end{equation}
The right-hand side is \(J(\kappa,\xi,\tau)=\xi^2G(\kappa,\tau)\). Combining
\eqref{eq:average_logvol_mean_identification}, \eqref{eq:atm_jensen_approximation}, and
\eqref{eq:jensen_term_identification} yields
\[
\log\mathcal L_t
\approx
\theta
+
A(\kappa,\tau)(\log\sigma_t-\theta)
+
J(\kappa,\xi,\tau).
\]
This proves \eqref{eq:level_identification}. Since \(A(\kappa,\tau)>0\) for \(\kappa>0\) and
\(\tau>0\), the equation can be inverted to obtain \eqref{eq:sigma_identified}.

\medskip
\noindent\textbf{Step 2: The skew and the leverage--vol-of-vol product (leading order).}

The additive skew measures the leading asymmetry of the implied-volatility surface around the
ATM point. In a stochastic-volatility model, this asymmetry is generated by the covariation between
the futures return shock and the volatility shock: if \(\rho=0\), futures shocks and volatility
shocks are locally uncorrelated and the first-order tilt of the smile vanishes; if \(\rho\neq0\), an
innovation in the futures price shifts the conditional distribution of future volatility, producing
an asymmetric risk-neutral distribution and hence a nonzero implied-volatility skew.

Appendix~\ref{app:identification-proof} derives this relation from the SDE by a first-order
expansion of the log-return $Y_\tau=\log(F_{t+\tau}/F_t)$ in the small parameter $\xi$: writing
$X_{t+u}=\theta+\xi M_u$ with $M_u=\int_0^u e^{-\kappa(u-s)}dB_s^{(2)}$ an $O(1)$ process, the third
cumulant of $Y_\tau$ is shown there to vanish at order $\xi^0$ (because the leading term is Gaussian)
and to be $O(\rho\xi)$ at first order, with a kernel of the form
\[
\frac{2}{\tau}\int_0^\tau\int_0^u e^{-\kappa(u-s)}\,ds\,du.
\]
The inner integral is
\[
\int_0^u e^{-\kappa(u-s)}\,ds
=
\frac{1-e^{-\kappa u}}{\kappa},
\]
so that
\begin{align}
\frac{2}{\tau}
\int_0^\tau
\int_0^u e^{-\kappa(u-s)}\,ds\,du
&=
\frac{2}{\tau}
\int_0^\tau
\frac{1-e^{-\kappa u}}{\kappa}\,du \nonumber\\
&=
\frac{2}{\tau}
\left\{
\frac{\tau}{\kappa}
-
\frac{1-e^{-\kappa\tau}}{\kappa^2}
\right\}
=f(\kappa,\tau).
\end{align}
Since the leading-order asymmetry is proportional to the local volatility
level \(\sigma_t\), the volatility-of-volatility \(\xi\), and the leverage correlation \(\rho\), and
since the standard Gram--Charlier link between the third cumulant of the log-return and the
linear-in-moneyness coefficient of the implied-volatility smile \citep{BackusForesiLiWu1997}
converts this cumulant into a skew of the observed additive form, the leading skew relation is
\[
\mathcal S_t
\approx
\rho\xi\sigma_t f(\kappa,\tau).
\]
This proves \eqref{eq:skew_identification} to leading order in \(\xi\). Inverting the relation gives
\[
\rho\xi
\approx
\frac{\mathcal S_t}{\sigma_t f(\kappa,\tau)}.
\]
This proves \eqref{eq:rho_xi_identified}.

\medskip
\noindent\textbf{Step 3: The convexity and volatility-of-volatility (leading order).}

The convexity ratio measures the curvature of the implied-volatility smile. In a stochastic-volatility
model, smile curvature is mainly a second-order effect, generated by uncertainty about future
volatility. Under the log-OU specification, this uncertainty is summarized by the conditional
variance of the average future log-volatility, already computed in Step~1:
\[
\frac12\operatorname{Var}_t(\bar X_{t,\tau})
=
\xi^2G(\kappa,\tau).
\]
Appendix~\ref{app:identification-proof} shows, again via the cumulant expansion, that the
fourth cumulant of \(Y_\tau\) is \(O(\xi^2)\) at leading order and, through the analogous
Gram--Charlier link between the fourth cumulant and the quadratic-in-moneyness (convexity)
coefficient of the smile, that the model-implied convexity ratio expands around a flat smile as
\[
\mathcal C_t
\approx
\exp\left\{\xi^2G(\kappa,\tau)\right\}.
\]
Since \(\xi^2G(\kappa,\tau)\) is small at the 30-day horizon, the first-order expansion
\(e^z-1\approx z\) gives
\[
\mathcal C_t-1
\approx
\xi^2G(\kappa,\tau).
\]
This proves \eqref{eq:convexity_identification}. Since \(G(\kappa,\tau)>0\) for
\(\kappa>0\) and \(\tau>0\), the equation can be inverted:
\[
\xi^2
\approx
\frac{\mathcal C_t-1}{G(\kappa,\tau)}.
\]
This proves \eqref{eq:xi_identified} and completes the proof.
\end{proof}

\begin{remark}[Why the relations are approximations]
\label{rem:why_approximation}
The relations in Theorem~\ref{thm:surface_identification} are not written as exact equalities
because the CVOL indicators are functions of the implied-volatility surface, not direct observations
of \(\sigma_t\), \(\xi\), and \(\rho\). The log-OU conditional moments used in Step~1 of the proof are
exact: the conditional mean of \(X_{t+u}\) and the variance of the average
\(\bar X_{t,\tau}\) follow exactly from the OU dynamics. Steps~2 and~3 are genuine leading-order
asymptotic results in the small-vol-of-vol regime, established via the cumulant expansion in
Appendix~\ref{app:identification-proof}, rather than exact identities. Higher-order smile terms,
finite-maturity effects, jump-risk compensation, volatility risk premia, and microstructure noise
are absorbed in the approximation error. For this reason, \(\xi_t^{\mathrm{surf}}\) and
\(\rho_t^{\mathrm{surf}}\) should be interpreted as effective surface-implied risk-neutral quantities
rather than exact physical diffusion parameters.
\end{remark}

\begin{remark}[Exact model-implied identification]
\label{rem:exact_model_implied_identification}
An exact version of the identification problem can be defined through the full stochastic-volatility
pricing map. Let
\[
I_t^{SV}(k,\tau;\sigma_t,\xi,\rho)
\]
denote the Black--76 implied volatility obtained from the log-OU stochastic-volatility price, where
\[
k=\log(K/F_t)
\]
is log-moneyness. The model-implied ATM level, skew, and convexity can then be defined directly
from the implied-volatility curve:
\[
\mathcal L_t^{SV}
=
I_t^{SV}(0,\tau;\sigma_t,\xi,\rho),
\]
\[
\mathcal S_t^{SV}
=
\left.
\frac{\partial I_t^{SV}(k,\tau;\sigma_t,\xi,\rho)}
{\partial k}
\right|_{k=0},
\]
and
\[
\mathcal C_t^{SV}
=
\left.
\frac{\partial^2 I_t^{SV}(k,\tau;\sigma_t,\xi,\rho)}
{\partial k^2}
\right|_{k=0},
\]
or, equivalently, by using the exact finite-difference definitions that match the CVOL construction.
Exact identification would then require solving the nonlinear system
\[
\mathcal L_t^{SV}(\sigma_t,\xi,\rho)=\mathcal L_t,
\qquad
\mathcal S_t^{SV}(\sigma_t,\xi,\rho)=\mathcal S_t,
\qquad
\mathcal C_t^{SV}(\sigma_t,\xi,\rho)=\mathcal C_t.
\]
This formulation is exact within the chosen stochastic-volatility model. However, it generally has
no closed-form inverse because \(I_t^{SV}\) itself is obtained from the Feynman--Kac expectation
or from the pricing PDE. The closed-form formulas in Theorem~\ref{thm:surface_identification}
should therefore be understood as an analytical leading-order approximation to this exact
nonlinear identification problem.
\end{remark}

\begin{remark}[Convexity close to the flat-smile boundary]
\label{rem:convexity_floor}
The theoretical leading-order relation \(\mathcal C_t-1\approx\xi^2G(\kappa,\tau)\) is
non-negative in the pure-diffusion model, because \(G(\kappa,\tau)>0\) and \(\xi^2\geq0\). In the
data, however, \(\mathcal C_t\) can be slightly below one: its sample minimum is 0.997
(Table~\ref{tab:summary_stats}). This occurs near the flat-smile boundary and can also be caused
by rounding in the reported CVOL fields. For empirical recovery of the surface-implied vol-of-vol
we therefore use the positive part \((\mathcal C_t-1)_+ := \max\{\mathcal C_t-1,0\}\) in
\eqref{eq:xi_surf_state} below. This correction does not change the theory; it only prevents a
nearly flat or slightly rounded convexity observation from producing an imaginary
volatility-of-volatility estimate.
\end{remark}

\subsection{Surface-Implied Vol-of-Vol and Leverage}
\label{subsec:surface_implied_states}

The identification system has an important implication. It not only provides a test of the constant-parameter log-OU model. It also shows how the implied-volatility surface can be used to construct time-varying pricing states. In the benchmark model, the volatility-of-volatility parameter \(\xi\) and the leverage coefficient \(\rho\) are constant. However, the CVOL convexity and skew indicators change over time. Once the identification equations are applied day by day, the recovered vol-of-vol and leverage also become time-varying. From Theorem~\ref{thm:surface_identification}, the convexity equation gives
\begin{equation}
\mathcal C_t-1
\approx
\xi^2G(\kappa,\tau).
\end{equation}
Therefore, for each date \(t\), the surface-implied effective vol-of-vol is
\begin{equation}
\xi_t^{\mathrm{surf}}
=
\left[
\frac{(\mathcal C_t-1)_+}{G(\kappa,\tau)}
\right]^{1/2},
\qquad
(\mathcal C_t-1)_+ := \max\{\mathcal C_t-1,\,0\},
\label{eq:xi_surf_state}
\end{equation}
using the positive-part correction of Remark~\ref{rem:convexity_floor}. Since \(\mathcal C_t\) is stochastic, equation~\eqref{eq:xi_surf_state} defines a stochastic vol-of-vol state. This is not an additional latent factor introduced by assumption. It is recovered directly from the observed smile curvature.

Similarly, the skew equation gives
\begin{equation}
\mathcal S_t
\approx
\rho\xi\sigma_t f(\kappa,\tau).
\end{equation}
Using the surface-implied value \(\xi_t^{\mathrm{surf}}\), we obtain the surface-implied leverage coefficient
\begin{equation}
\rho_t^{\mathrm{surf}}
=
\frac{\mathcal S_t}
{\widetilde\xi_t^{\mathrm{surf}}\sigma_t f(\kappa,\tau)},
\qquad
\widetilde\xi_t^{\mathrm{surf}} := \max\{\xi_t^{\mathrm{surf}},\,\xi_{\min}\},
\label{eq:rho_surf_state}
\end{equation}
where \(\xi_{\min}>0\) is a small numerical floor used only to avoid division by zero on days when
the smile is essentially flat (equivalently, when \((\mathcal C_t-1)_+=0\)). Thus, the skew indicator
identifies a time-varying effective correlation between futures shocks and volatility shocks. The
level, convexity, and skew indicators therefore generate the three main pricing states:
\begin{equation}
\mathcal L_t \longrightarrow \sigma_t,
\qquad
\mathcal C_t \longrightarrow \xi_t^{\mathrm{surf}},
\qquad
\mathcal S_t \longrightarrow \rho_t^{\mathrm{surf}}.
\label{eq:surface_state_map}
\end{equation}

This is the main bridge between the empirical surface analysis and the pricing model. The constant-\(\xi\), constant-\(\rho\) log-OU specification is best understood as a benchmark. It is useful because it gives a simple model with a clear pricing PDE and a transparent identification structure. But the data suggest that soybean option smiles contain more information than a constant-parameter model can absorb. The convexity factor changes over time, so the smile-implied vol-of-vol changes over time. The skew also changes over time, so the leverage channel changes as well.

It is important to interpret \(\xi_t^{\mathrm{surf}}\) and \(\rho_t^{\mathrm{surf}}\) carefully. They are not necessarily the physical parameters of the volatility process. They are surface-implied effective quantities under the pricing measure. They may contain true stochastic volatility-of-volatility, volatility risk premia, jump-risk compensation, and temporary demand for tail protection. For this reason, we do not use the raw values mechanically in the pricing section. Instead, we treat them as observable surface-driven state variables and regularize them before using them as risk-neutral pricing inputs.

This interpretation also clarifies the empirical gap between the time-series estimate and the surface-implied vol-of-vol. The log-OU fit to the ATM level gives
\begin{equation}
\widehat{\xi}_{\mathrm{TS}}=0.842,
\end{equation}
while the average convexity-implied value, computed from \eqref{eq:xi_surf_state} with $\kappa=6.651$
(the log-OU estimate for the level) and $\tau=30/252$, is
\begin{equation}
\overline{\xi}_{\mathrm{surf}}=2.543.
\end{equation}
This difference should not be read as a pure failure of the physical-measure model. The time-series estimate and the surface-implied quantity live in different informational and, potentially, measure-theoretic environments. Under the identifying restriction that one constant vol-of-vol parameter should describe both, the gap is large; economically, it may reflect time-varying vol-of-vol, jump risk, a volatility risk premium, or residual approximation error. The key implication for the pricing exercise is that smile curvature supplies information not contained in the ATM level alone. Appendix~\ref{app:spec-test} formalizes the equality restriction as a HAC-based diagnostic (Newey--West $z\approx 45$ at a bandwidth of 20 trading days, robust to bandwidths from 5 to 100 days).

The skew equation gives a second diagnostic. If one uses the smaller time-series value \(\widehat{\xi}_{\mathrm{TS}}\) in the skew identification equation, the implied correlation can exceed one in magnitude, which is not admissible. By contrast, using \(\xi_t^{\mathrm{surf}}\) gives leverage estimates that are economically more reasonable. This supports the view that convexity and skew should be read together. Convexity identifies the magnitude of volatility uncertainty, while skew identifies how that uncertainty is linked to futures returns.

The pricing model developed in the next section follows this logic. At each pricing date, we recover \(\sigma_t\), \(\xi_t^{\mathrm{surf}}\), and \(\rho_t^{\mathrm{surf}}\) from the CVOL surface. The raw surface-implied quantities are then transformed into regularized risk-neutral inputs,
\begin{equation}
\xi_t^Q
=
\mathcal R_{\xi}\left(\xi_t^{\mathrm{surf}}\right),
\qquad
\rho_t^Q
=
\mathcal R_{\rho}\left(\rho_t^{\mathrm{surf}}\right),
\label{eq:regularized_pricing_states}
\end{equation}
where \(\mathcal R_{\xi}\) and \(\mathcal R_{\rho}\) are simple stability maps. For example, \(\xi_t^Q\) may be capped above to avoid using jump-risk or risk-premium compensation as a pure diffusion coefficient, and \(\rho_t^Q\) is projected into an admissible interval inside \((-1,1)\). The resulting model is a surface-driven stochastic-volatility model with stochastic vol-of-vol and stochastic leverage:
\begin{align}
\frac{dF_t}{F_t}
&=
\sigma_t\,dW_t^F,
\label{eq:surface_driven_futures}\\
d\log\sigma_t
&=
\kappa(\theta-\log\sigma_t)dt
+
\xi_t^Q\,dW_t^\sigma,
\label{eq:surface_driven_logvol}\\
d\langle W^F,W^\sigma\rangle_t
&=
\rho_t^Q\,dt.
\label{eq:surface_driven_corr}
\end{align}
In this form, the implied-volatility surface is no longer only an output of the model. It becomes an input that updates the volatility level, the vol-of-vol channel, and the leverage channel.

\section{Surface-Driven Option Pricing and Numerical Validation}
\label{sec:pricing_numerics_validation}
\label{sec:pricing_numerics}

The previous section shows how the CME CVOL soybean surface can be converted into three pricing inputs: a volatility state, an effective vol-of-vol state, and an effective leverage state. We now use these quantities to price soybean futures options. The numerical exercise has a simple structure. Monte Carlo simulation is the benchmark, because it prices the risk-neutral model directly from the Feynman--Kac representation. The finite-difference method is used as an independent validation of the pricing PDE. A neural network is then trained as a fast pricer. It is not a residual-correction model; it learns the pricing map produced by the stochastic-volatility model.

This section is organized as follows. Subsection~\ref{subsec:pricing_state_recovery} describes the surface states used for pricing. Subsection~\ref{subsec:pricing_model} gives the risk-neutral model. Subsection~\ref{subsec:numerical_methods} explains the Monte Carlo, finite-difference, and neural-network procedures. Subsection~\ref{subsec:event_impact} studies the COVID and La~Ni\~na event windows. Subsection~\ref{subsec:numerical_validation} reports the numerical accuracy, the full-sample relative-error table, and the runtime comparison.

\subsection{Surface states used in the pricing experiment}
\label{subsec:pricing_state_recovery}

At each pricing date, the observed surface indicators are transformed into the state vector
\[
(\mathcal L_t,\mathcal S_t,\mathcal C_t)
\quad \longmapsto \quad
(\sigma_t,\xi_t^Q,\rho_t^Q).
\]
The volatility state \(\sigma_t\) is recovered from the ATM level. The quantity \(\xi_t^Q\) is derived from convexity and is interpreted as an effective risk-neutral vol-of-vol. The coefficient \(\rho_t^Q\) is obtained from skew and measures the effective leverage between futures returns and volatility shocks. These quantities are surface-implied. They may contain risk premia, jump-risk compensation, and short-lived demand for tail protection. For this reason, the raw surface-implied values are regularized before they are used as diffusion inputs.

Table~\ref{tab:surface_states_regime} summarizes the recovered states for the full sample and for the two event windows. The La~Ni\~na window is the most stressed agricultural period in the sample. Its mean ATM volatility is \(25.29\%\), compared with \(18.60\%\) in the full sample and \(17.64\%\) during COVID. The recovered volatility state is also much higher during La~Ni\~na. The mean additive skew rises from \(1.96\) volatility points in the full sample to \(4.83\) volatility points during La~Ni\~na. This is important because it shows that the event is not only a high-volatility period. It is also a strong smile period.

\begin{table}[H]
\centering
\caption{Surface-implied pricing states by regime. The skew column reports the additive skew in volatility points.}
\label{tab:surface_states_regime}
\renewcommand{\arraystretch}{1.15}
\begin{tabular}{lrrrrrrrr}
\toprule
Regime & Days & ATM (\%) & \(\sigma_t\) & \(\xi_t^{\mathrm{surf}}\) & \(\xi_t^Q\) & \(\rho_t^Q\) & \(\mathcal C_t\) & Skew \\
\midrule
Full sample & 2998 & 18.60 & 0.188 & 2.54 & 1.20 & 0.394 & 1.077 & 1.96 \\
COVID & 213 & 17.64 & 0.173 & 2.65 & 1.19 & 0.658 & 1.086 & 3.16 \\
La Ni\~na & 85 & 25.29 & 0.292 & 2.78 & 1.20 & 0.642 & 1.090 & 4.83 \\
\bottomrule
\end{tabular}
\end{table}

\FloatBarrier

\subsection{Risk-neutral surface-driven pricing model}
\label{subsec:pricing_model}

Let \(F_t\) denote the soybean futures price and let \(X_t=\log\sigma_t\) denote log-volatility. Under the risk-neutral measure \(\mathbb Q\), the pricing model is
\begin{align}
\frac{dF_t}{F_t}
&=
\sigma_t\,dW_t^F,
\label{eq:rn_surface_futures}\\
dX_t
&=
\kappa(\theta-X_t)dt+\xi_t^Q\,dW_t^\sigma,
\qquad X_t=\log\sigma_t,
\label{eq:rn_surface_logvol}\\
d\langle W^F,W^\sigma\rangle_t
&=
\rho_t^Qdt.
\label{eq:rn_surface_correlation}
\end{align}
The mean-reversion parameters \(\kappa\) and \(\theta\) come from the log-OU fit to the ATM volatility level. The quantities \(\xi_t^Q\) and \(\rho_t^Q\) come from the current implied-volatility surface. They are held fixed over the life of each option. This frozen-surface convention keeps the pricing problem two-dimensional, with state variables \((F_t,X_t)\). A fully dynamic version with stochastic \(\xi_t^Q\) and stochastic \(\rho_t^Q\) would require additional state variables and is left for future work.

The raw surface-implied vol-of-vol is computed from convexity as
\begin{equation}
\xi_t^{\mathrm{surf}}
=
\left[
\frac{(\mathcal C_t-1)^+}{G(\kappa,\tau)}
\right]^{1/2},
\qquad
(\mathcal C_t-1)^+=\max\{\mathcal C_t-1,0\}.
\label{eq:xi_surf_section4}
\end{equation}
The positive part is used because the observed convexity can be very close to one. A value slightly below one is interpreted as a near-flat smile rather than a negative vol-of-vol. The leverage state is recovered from the skew equation,
\begin{equation}
\rho_t^{\mathrm{surf}}
=
\frac{\mathcal S_t}
{\widetilde\xi_t^{\mathrm{surf}}\sigma_t f(\kappa,\tau)},
\qquad
\widetilde\xi_t^{\mathrm{surf}}=\max\{\xi_t^{\mathrm{surf}},\xi_{\min}\}.
\label{eq:rho_surf_section4}
\end{equation}
The pricing inputs are then regularized by
\begin{align}
\xi_t^Q
&=\min\{\xi_t^{\mathrm{surf}},\xi^Q_{\max}\},
\label{eq:xi_regularized_section4}\\
\rho_t^Q
&=\Pi_{[-\rho_{\max},\rho_{\max}]}
\left(\rho_t^{\mathrm{surf}}\right),
\label{eq:rho_regularized_section4}
\end{align}
where \(\Pi\) denotes projection. This step prevents very large smile-implied values, which may reflect jumps or risk premia, from being used mechanically as pure diffusion coefficients.

\subsection{Pricing equation and numerical methods}
\label{subsec:numerical_methods}

Let \(V(t,F,x)\) be the value of a European option with payoff \(\Phi(F_T)\), where \(x=\log\sigma\). Conditional on a pricing date, write \(\xi=\xi_t^Q\) and \(\rho=\rho_t^Q\). The frozen-surface model gives the pricing PDE
\begin{equation}
0
=
V_t
+
\frac12 e^{2x}F^2V_{FF}
+
\rho\xi e^x F V_{Fx}
+
\frac12\xi^2V_{xx}
+
\kappa(\theta-x)V_x
-rV,
\label{eq:surface_pricing_pde_section4}
\end{equation}
with terminal condition \(V(T,F,x)=\Phi(F)\). The mixed derivative term is caused by the correlation between futures shocks and volatility shocks. Its coefficient is \(\rho\xi e^x\), not \(\rho\xi e^{2x}\), because the state variable in the second dimension is \(x=\log\sigma\).

The same price has the Feynman--Kac representation
\begin{equation}
V(t,F,x)
=
e^{-r(T-t)}
\mathbb E^{\mathbb Q}
\left[\Phi(F_T)\mid F_t=F,\;X_t=x\right].
\label{eq:feynman_kac_section4}
\end{equation}
This representation is the basis for Monte Carlo simulation. If \(\Delta\) is the time step, the log-volatility state is simulated exactly:
\begin{equation}
X_{n+1}
=
\theta+(X_n-\theta)e^{-\kappa\Delta}
+
\xi_t^Q
\sqrt{\frac{1-e^{-2\kappa\Delta}}{2\kappa}}
Z_{2,n},
\label{eq:mc_exact_logvol_section4}
\end{equation}
where \(Z_{2,n}\sim N(0,1)\). The futures price is updated by
\begin{equation}
\log F_{n+1}
=
\log F_n
-
\frac12e^{2X_n}\Delta
+
e^{X_n}\sqrt{\Delta}\,Z_{1,n},
\label{eq:mc_futures_update_section4}
\end{equation}
with \(\mathrm{Corr}(Z_{1,n},Z_{2,n})=\rho_t^Q\). The Monte Carlo benchmark is
\begin{equation}
V^{MC}
=
e^{-rT}\frac1M\sum_{m=1}^M\Phi(F_T^{(m)}).
\label{eq:mc_estimator_section4}
\end{equation}
Antithetic variates are used in the simulation. Black--76 is used only to convert model prices into implied-volatility units. It is not used as the pricing benchmark.

The finite-difference method solves the PDE independently of Monte Carlo. We use the grid variables
\[
y=\log(F/K),
\qquad
x=\log\sigma.
\]
If \(V(t,F,x)=U(t,y,x)\), then
\[
FV_F=U_y,
\qquad
F^2V_{FF}=U_{yy}-U_y,
\qquad
FV_{Fx}=U_{yx}.
\]
The PDE in time-to-maturity form is therefore
\begin{equation}
U_\tau
=
\frac12e^{2x}(U_{yy}-U_y)
+
\rho\xi e^xU_{yx}
+
\frac12\xi^2U_{xx}
+
\kappa(\theta-x)U_x
-rU.
\label{eq:log_grid_pde_section4}
\end{equation}
We use a Crank--Nicolson scheme with Rannacher smoothing. The first smoothing steps are helpful because the option payoff has a kink at the strike. The \(y\)-domain is chosen adaptively for each option based on the current moneyness, maturity, and volatility. This avoids using a fixed truncation interval for all strikes and maturities.

The neural network is trained after numerical labels have been generated. Its input vector is
\[
z=
\left(
\log(F_0/K),\;T,\;X_0,\;\xi_t^Q,\;\rho_t^Q,\;\kappa,\;\theta,\;r
\right),
\]
and the output is the normalized price \(V/K\). The training labels come from Monte Carlo prices. The design is regime-balanced, so that the training set contains normal, COVID, and La~Ni\~na states. This is important because a fast pricer trained only on calm periods may extrapolate poorly in stressed markets.

\subsection{Event-period pricing impact}
\label{subsec:event_impact}

We compare the full sample with two event windows. The COVID window covers March--December 2020 and represents a broad market and supply-chain disruption. The La~Ni\~na window spans May--August 2021 and represents a weather-risk episode more closely linked to soybean fundamentals. Table~\ref{tab:event_price_iv_impact} reports the Monte Carlo price and implied-volatility impact of each event period relative to the full-sample benchmark. In this experiment, COVID has a lower mean option price and lower mean implied volatility. The mean Monte Carlo price is \(12.45\%\) below the full-sample level, and the mean implied volatility is lower by about \(131\) basis points. La~Ni\~na has the opposite effect. The mean Monte Carlo price is \(37.66\%\) above the full-sample level, and the mean implied volatility is higher by about \(298\) basis points.

\begin{table}[H]
\centering
\caption{Monte Carlo price and implied-volatility impact by event period. The full sample is the reference period.}
\label{tab:event_price_iv_impact}
\renewcommand{\arraystretch}{1.15}
\begin{tabular}{lrrrrr}
\toprule
Regime & Contracts & Mean MC price & Price impact (\%) & Mean MC IV (\%) & IV impact (bps) \\
\midrule
COVID & 45 & 58.57 & -12.45 & 18.82 & -130.76 \\
La Ni\~na & 45 & 92.08 & 37.66 & 23.11 & 298.30 \\
\bottomrule
\end{tabular}
\end{table}
\begin{flushleft}
\footnotesize
\textit{Note:} The column ``Contracts'' does not refer to calendar days.
For each event period, prices are computed on five representative pricing
dates, three moneyness levels, and three maturities. This gives
\(5\times 3\times 3=45\) option contracts per event window.
\end{flushleft}
Figure~\ref{fig:mc_iv_smiles_regime} gives the implied-volatility smiles by regime. The La~Ni\~na smile is above the full-sample and COVID smiles at the 30-, 90-, and 180-day maturities. The gap is strongest at 30 days, but it remains visible at longer maturities. This confirms that the event effect is not only a shift in the ATM level. The surface shape also changes, especially on the call side of the smile.

\begin{figure}[H]
\centering
\includegraphics[width=\textwidth]{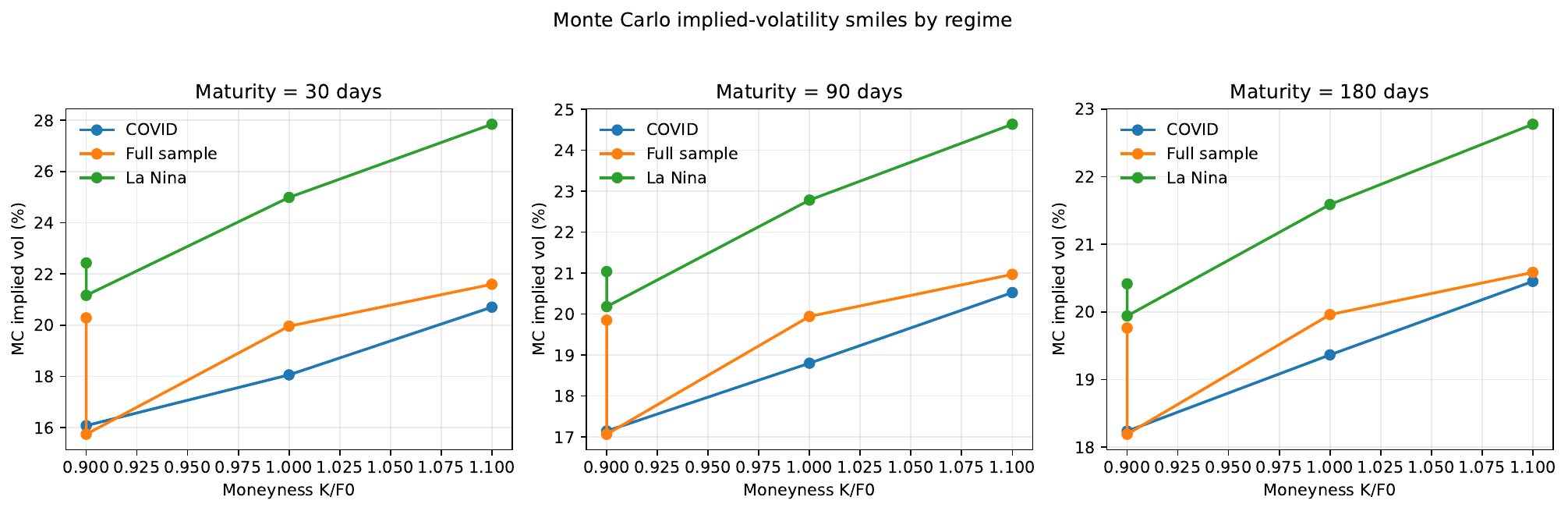}
\caption{Monte Carlo implied-volatility smiles by regime. La~Ni\~na produces the highest smile across the three maturities.}
\label{fig:mc_iv_smiles_regime}
\end{figure}

\FloatBarrier

\subsection{Full-sample numerical validation}
\label{subsec:numerical_validation}

Figure~\ref{fig:full_sample_price_grid} compares Monte Carlo, finite differences, and the neural network on the same full-sample price grid. The three curves almost overlap at 30, 90, and 180 days. This is the first check that the numerical methods are internally consistent: Monte Carlo gives the benchmark, finite differences solve the PDE, and the neural network reproduces the same pricing map after training.

\begin{figure}[H]
\centering
\includegraphics[width=\textwidth]{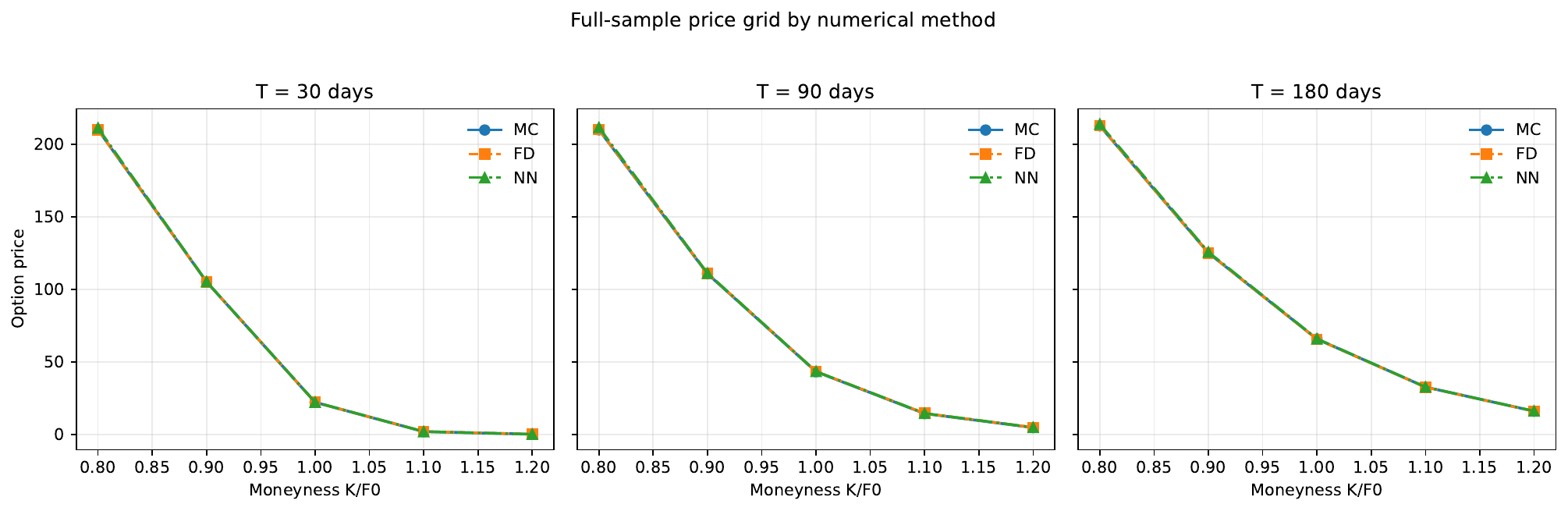}
\caption{Full-sample price grid by numerical method. Monte Carlo, finite differences, and the neural network are plotted on the same moneyness grid.}
\label{fig:full_sample_price_grid}
\end{figure}

The visual comparison in Figure~\ref{fig:full_sample_price_grid} is useful, but price curves can hide small differences when the option is deep in or out of the money. Table~\ref{tab:full_sample_price_relative_errors} therefore reports the relative price errors directly. For most core strikes, the errors are small. At 30 days and \(K/F_0=1.20\), the relative error is larger because the Monte Carlo benchmark price is only \(0.146\). In that case, even a very small absolute price difference becomes a large percentage error. This is why relative price errors must be read together with the underlying Monte Carlo price.

\begin{table}[H]
\centering
\caption{Full-sample price relative errors. Monte Carlo is the benchmark. Relative errors are computed as \(100(V^{\mathrm{method}}-V^{MC})/V^{MC}\).}
\label{tab:full_sample_price_relative_errors}
\scriptsize
\resizebox{\textwidth}{!}{%
\begin{tabular}{rrrrrrrrrr}
\toprule
Maturity & \(K/F_0\) & MC price & FD price & FD rel. err. (\%) & FD abs. rel. err. (\%) & NN price & NN rel. err. (\%) & NN abs. rel. err. (\%) & MC s.e. \\
\midrule
30 & 0.80 & 209.915 & 209.909 & -0.003 & 0.003 & 211.111 & 0.570 & 0.570 & 0.105 \\
30 & 0.90 & 105.122 & 105.118 & -0.003 & 0.003 & 105.268 & 0.139 & 0.139 & 0.104 \\
30 & 1.00 & 22.152 & 22.150 & -0.011 & 0.011 & 22.233 & 0.366 & 0.366 & 0.070 \\
30 & 1.10 & 1.961 & 2.034 & 3.704 & 3.704 & 1.968 & 0.333 & 0.333 & 0.023 \\
30 & 1.20 & 0.146 & 0.175 & 19.822 & 19.822 & 0.175 & 19.768 & 19.768 & 0.006 \\
90 & 0.80 & 209.880 & 210.048 & 0.080 & 0.080 & 211.476 & 0.760 & 0.760 & 0.209 \\
90 & 0.90 & 110.707 & 110.910 & 0.183 & 0.183 & 111.059 & 0.318 & 0.318 & 0.195 \\
90 & 1.00 & 43.221 & 43.373 & 0.352 & 0.352 & 43.373 & 0.351 & 0.351 & 0.147 \\
90 & 1.10 & 14.387 & 14.568 & 1.258 & 1.258 & 14.524 & 0.951 & 0.951 & 0.093 \\
90 & 1.20 & 4.724 & 4.858 & 2.844 & 2.844 & 4.861 & 2.904 & 2.904 & 0.056 \\
180 & 0.80 & 212.651 & 212.728 & 0.036 & 0.036 & 213.713 & 0.500 & 0.500 & 0.312 \\
180 & 0.90 & 125.061 & 125.061 & -0.001 & 0.001 & 125.639 & 0.462 & 0.462 & 0.282 \\
180 & 1.00 & 65.735 & 65.585 & -0.227 & 0.227 & 65.923 & 0.286 & 0.286 & 0.228 \\
180 & 1.10 & 32.646 & 32.596 & -0.152 & 0.152 & 32.689 & 0.131 & 0.131 & 0.173 \\
180 & 1.20 & 16.046 & 16.036 & -0.062 & 0.062 & 16.109 & 0.394 & 0.394 & 0.127 \\
\bottomrule
\end{tabular}}
\end{table}

The largest average relative errors occur at the 30-day maturity because the far out-of-the-money option is very cheap. At 90 and 180 days, both finite differences and the neural network remain close to the Monte Carlo results. At 180 days, the mean absolute relative error is about \(0.10\%\) for finite differences and \(0.35\%\) for the neural network.

Figure~\ref{fig:fd_vs_mc_full_sample} gives a second validation of the finite-difference solver. The finite-difference prices lie close to the 45-degree line against Monte Carlo prices. This confirms that the corrected log-grid PDE and the Monte Carlo engine produce consistent prices on the full-sample validation grid.

\begin{figure}[H]
\centering
\includegraphics[width=0.78\textwidth]{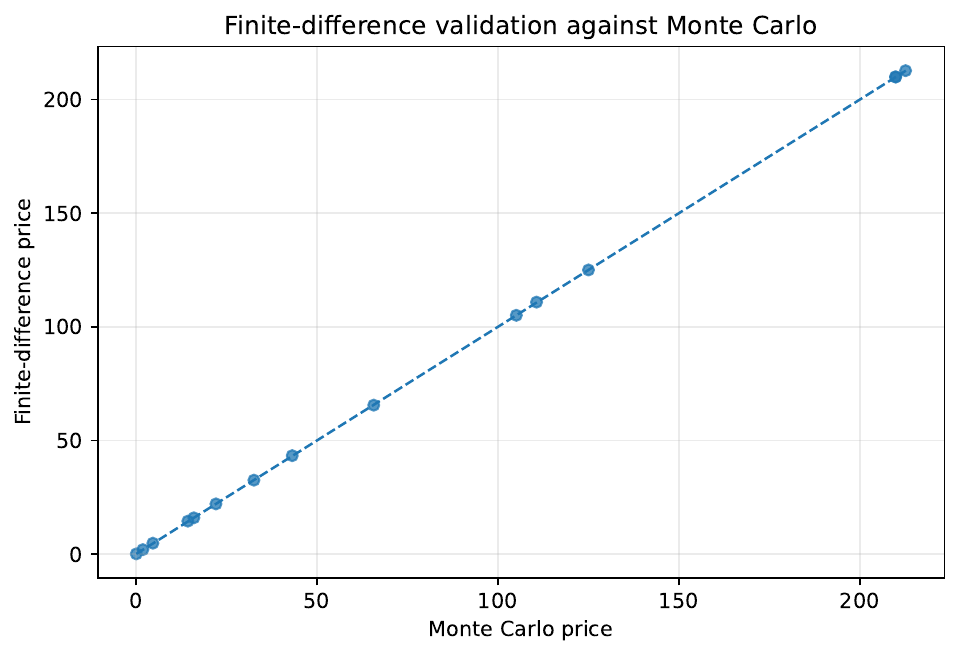}
\caption{Finite-difference validation against the Monte Carlo benchmark on the full-sample grid.}
\label{fig:fd_vs_mc_full_sample}
\end{figure}

The finite-difference method is not run on every daily observation. It is used on a focused full-sample validation grid with five moneyness levels and three
maturities, giving 15 contracts. This keeps the PDE validation transparent and
computationally feasible. Monte Carlo remains the benchmark for the broader
pricing experiment, while the neural network is evaluated on a larger
regime-balanced test set.

\subsection{Neural-network fast pricer}
\label{subsec:nn_fast_pricer_results}

The neural network is evaluated on held-out Monte Carlo labels from the full sample, COVID, and La~Ni\~na regimes. Figure~\ref{fig:nn_vs_mc_regime} plots the neural-network prices against the Monte Carlo labels. The points are close to the 45-degree line in all three panels. This shows that the network has learned the pricing map well over both normal and stressed surface states.

\begin{figure}[H]
\centering
\includegraphics[width=\textwidth]{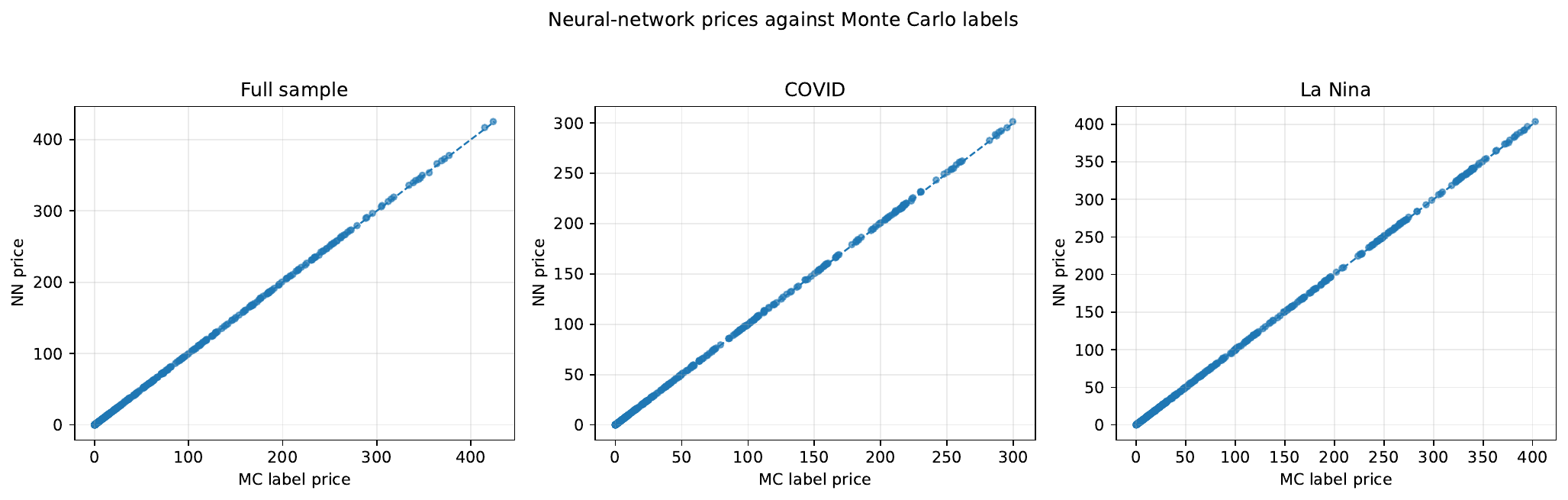}
\caption{Neural-network prices against Monte Carlo labels on the test set.}
\label{fig:nn_vs_mc_regime}
\end{figure}

Finally, Table~\ref{tab:runtime_speedup} reports the computational cost. Monte Carlo is the benchmark, but it is also the slowest method. Finite differences are of the same order of cost on the focused validation grid. Once the network is trained, neural-network inference is essentially instantaneous. In this run, neural-network inference takes approximately \(1.9\times 10^{-4}\) seconds per option, corresponding to a speed-up of about \(5.1\times 10^3\) relative to Monte Carlo on the same grid. This is the practical value of the fast-pricer layer. For a fair timing comparison, all three methods are evaluated on the same full-sample validation grid. This grid contains five moneyness levels and three
maturities, for a total of 15 contracts. The neural-network timing reports only
inference after training, since the trained network is intended to be reused for
many pricing queries.
\begin{table}[H]
\centering
\caption{Runtime comparison on the full-sample validation grid. The grid contains
five moneyness levels and three maturities, giving 15 timed contracts. Neural-network time refers to inference after training.}
\label{tab:runtime_speedup}
\begin{tabular}{lrrrr}
\toprule
Method & Timed contracts & Total time (sec.) & Time per contract (sec.) & Speed-up vs MC \\
\midrule
Monte Carlo & 15 & 14.277 & 0.9518 & 1.00 \\
Finite difference & 15 & 12.402 & 0.8268 & 1.15 \\
Neural-network inference & 15 & 0.0028 & 0.000188 & 5059.59 \\
\bottomrule
\end{tabular}
\end{table}

\section{Conclusion and Discussion}
\label{sec:conclusion}
This paper proposes a surface-implied stochastic-volatility framework for commodity options. The key message is that the implied-volatility surface is not only an output to be fitted; it can also be used as an input for model identification and model testing. The surface level identifies the latent volatility state, the skew identifies the leverage-vol-of-vol product, and the convexity identifies the smile-implied vol-of-vol. The empirical gap between time-series and smile-implied vol-of-vol, formalized as a specification test in Appendix~\ref{app:spec-test}, shows that a constant-parameter log-OU model is not rich enough for the full surface. A surface-driven time-varying vol-of-vol extension would give a natural next modeling step, and the numerical pricing and hedging results outlined in Section~\ref{sec:pricing_numerics_validation} and Appendix~\ref{app:numerics} illustrate this work in practice.

The framework is also timely for forward-looking agricultural risk. As of September 2026, the World Meteorological Organization and the U.S. Climate Prediction Center report a firmly established and strengthening El~Ni\~no, with a very high likelihood of persistence through the Northern Hemisphere winter of 2026--2027 and elevated odds of very strong intensity \citep{WMO2026ElNino,NOAA2026ENSO}. Such climate regimes can shift regional temperature and precipitation distributions, creating an economically relevant setting in which option-surface level, skew, and curvature may adjust before realized crop outcomes are known. We do not attempt to forecast soybean yields from ENSO in this paper; rather, the framework provides a way to translate evolving option-surface information into stochastic-volatility pricing states.

\appendix

\section{Proofs}
\label{app:proofs}

\subsection{Proof of the exact transition density formulas}
\label{app:transition-density-proofs}

\paragraph{OU.} The SDE $dX_t=\kappa(\theta-X_t)dt+\xi dW_t$ is linear with additive noise. Applying
It\^o's formula to $e^{\kappa t}X_t$ gives $d(e^{\kappa t}X_t)=\kappa\theta e^{\kappa t}dt+\xi
e^{\kappa t}dW_t$; integrating from $t$ to $t+h$ and rearranging yields the variation-of-constants
solution \eqref{eq:ou_solution}. The stochastic integral $\int_t^{t+h}e^{-\kappa(t+h-s)}dW_s$ is a
Wiener integral of a deterministic kernel, hence Gaussian with mean zero and, by It\^o isometry,
variance $\int_t^{t+h}e^{-2\kappa(t+h-s)}ds=(1-e^{-2\kappa h})/(2\kappa)$. This gives
\eqref{eq:ou_transition}--\eqref{eq:ou_moments}.

\paragraph{Log-OU.} Set $Y_t=\ln X_t$. By It\^o's formula applied to the CIR-free log transform
(here there is no diffusion term in $X$ to generate a second-derivative correction because the
model is specified directly for $Y_t=\ln X_t$, not derived from a diffusion in $X_t$), $Y_t$ solves
the OU SDE $dY_t=\kappa(\theta^{\log}-Y_t)dt+\xi dW_t$, so \eqref{eq:ou_transition} applies to $Y_t$
verbatim. The transition density of $X_{t+h}=e^{Y_{t+h}}$ follows by the change-of-variables formula
for densities, $p_X(x)=p_Y(\ln x)/x$, which gives the Jacobian factor $1/X_{t+h}$ in
\eqref{eq:logou_density}.

\paragraph{CIR.} The CIR transition can be obtained either by solving the Kolmogorov forward
(Fokker--Planck) equation via Laplace transform in $X_{t+h}$
or, more transparently, via the classical representation of a CIR process as a squared Euclidean
norm of correlated OU processes when the "dimension" $d:=4\kappa\theta/\xi^2$ is a positive integer,
extended to non-integer $d$ by analytic continuation. Concretely, if $d\in\{1,2,\dots\}$ and
$Z^{(1)}_t,\dots,Z^{(d)}_t$ are i.i.d.\ scalar OU processes with mean-reversion speed $\kappa/2$ and
stationary variance $\xi^2/(4\kappa)$, then $X_t:=\sum_{i=1}^d (Z^{(i)}_t)^2$ solves the CIR SDE in
Definition~\ref{def:candidates}. Each $Z^{(i)}_{t+h}\mid Z^{(i)}_t$ is Gaussian by the OU transition
law, so $X_{t+h}/(\text{scale})=\sum_i (Z^{(i)}_{t+h})^2/(\text{scale})$ is, conditional on $X_t$, a
sum of $d$ independent squared Gaussians with a common non-zero mean -- by definition, a
noncentral chi-squared random variable with $d$ degrees of freedom and noncentrality parameter
proportional to $X_t e^{-\kappa h}$. Writing out the noncentral chi-squared density and rescaling by
the constant $c$ in \eqref{eq:cir_uvq} produces exactly \eqref{eq:cir_density}. Because both sides
of the resulting identity are analytic in $d$ (equivalently, in $q=d/2-1$), the formula extends to
all $q>-1$, i.e.\ to all $\kappa,\theta,\xi>0$, not only those for which $d$ is an even integer. We
verified \eqref{eq:cir_density} independently by Monte Carlo simulation of the CIR SDE at several
parameter values spanning the range in Table~\ref{tab:mle_results} (both the "level" scale, where
$q\approx 30$--$60$, and the near-unity scale of the ratio and convexity indicators, where the
noncentrality parameter is large): the simulated transition density matched the closed-form
evaluation of \eqref{eq:cir_density} to within Monte Carlo error in every case.

\subsection{Cumulant-expansion derivation of the skew and convexity identification relations}
\label{app:identification-proof}

This appendix strengthens Steps~2--3 of the proof of Theorem~\ref{thm:surface_identification},
which in the main text are motivated by an integrated-kernel argument but not derived directly from
the SDE. We derive the leading order in $\xi$ of the third and fourth cumulants of the log-return
$Y_\tau=\log(F_{t+\tau}/F_t)$ directly from \eqref{eq:rn_futures}--\eqref{eq:rn_corr}, and invoke the
classical Gram--Charlier link between cumulants and the shape of the Black--76 implied-volatility
smile to connect these cumulants to skew and convexity.

\paragraph{Small-vol-of-vol expansion of the log-return.} Fix $t=0$ and condition on $X_0=\theta$
(the leading-order, at-the-mean case; away from the mean the same computation goes through with an
extra deterministic drift term that does not affect the cumulants below). Write
\[
X_u=\theta+\xi M_u,
\qquad
M_u:=\int_0^u e^{-\kappa(u-s)}\,dB_s^{(2)},
\]
so that $M_u$ is an $O(1)$, mean-zero Gaussian process not depending on $\xi$. Using
$\sigma_u=e^{X_u}=\bar\sigma\,e^{\xi M_u}=\bar\sigma\bigl(1+\xi M_u+O(\xi^2)\bigr)$ with
$\bar\sigma:=e^{\theta}$, and $Y_\tau=\int_0^\tau \sigma_u\,dB_u^{(1)}-\tfrac12\int_0^\tau
\sigma_u^2\,du$ (the usual It\^o correction for $d\log F=\sigma dB^{(1)}-\tfrac12\sigma^2dt$), a
first-order expansion in $\xi$ gives
\[
Y_\tau = Y_\tau^{(0)} + \xi\,Y_\tau^{(1)} + O(\xi^2),
\]
\[
Y_\tau^{(0)} = \bar\sigma B_\tau^{(1)} - \tfrac12\bar\sigma^2\tau,
\qquad
Y_\tau^{(1)} = \bar\sigma\int_0^\tau M_u\,dB_u^{(1)} - \bar\sigma^2\int_0^\tau M_u\,du.
\]
$Y_\tau^{(0)}$ is Gaussian (constant-volatility Black--Scholes log-return), so all of its cumulants
of order $3$ and higher vanish identically. Consequently, the third and fourth cumulants of $Y_\tau$
are driven entirely by the $\xi$- and $\xi^2$-order cross terms between $Y_\tau^{(0)}$ and
$Y_\tau^{(1)}$.

\paragraph{Third cumulant (skew).} Writing $\widetilde Y^{(0)}=Y_\tau^{(0)}-\mathbb E[Y_\tau^{(0)}]$
and using $\mathbb E[Y_\tau^{(1)}]=0$ (the first term is a mean-zero It\^o integral and the second is the time integral of the centered Gaussian process $M_u$),
\[
\kappa_3(Y_\tau)
= \mathbb E\bigl[(\widetilde Y^{(0)}+\xi Y_\tau^{(1)})^3\bigr]
= \underbrace{\mathbb E[(\widetilde Y^{(0)})^3]}_{=0\ (\text{Gaussian})}
+ 3\xi\,\mathbb E\bigl[(\widetilde Y^{(0)})^2\,Y_\tau^{(1)}\bigr] + O(\xi^2).
\]
Since $\widetilde Y^{(0)}=\bar\sigma B_\tau^{(1)}$ is already centered, and since for any two
jointly Gaussian, mean-zero random variables $A,B$ one has $\mathbb E[A^2B]=0$ (write
$B=\rho_{AB}(\sigma_B/\sigma_A)A+\sqrt{1-\rho_{AB}^2}\,\sigma_B\varepsilon$ with $\varepsilon$
standard normal and independent of $A$; then $\mathbb E[A^2B]=\rho_{AB}(\sigma_B/\sigma_A)\mathbb
E[A^3]=0$), the contribution of $-\bar\sigma^2\int_0^\tau M_u\,du$ to
$\mathbb E[(\widetilde Y^{(0)})^2Y_\tau^{(1)}]$ vanishes term by term, because $(B_\tau^{(1)},M_u)$
is jointly Gaussian for every fixed $u$. What remains is the Wiener-chaos contribution of
$\bar\sigma\int_0^\tau M_u\,dB_u^{(1)}$, which is the classical mechanism by which correlated
stochastic volatility generates return skewness: an It\^o-integral computation (using the product
formula for iterated Wiener integrals, as in the standard small-vol-of-vol expansions of
\citealp{FouquePapanicolaouSircar2000} and \citealp{Lewis2000}) gives
\[
\kappa_3(Y_\tau) = 3\rho\xi\bar\sigma^3\int_0^\tau\!\int_0^u e^{-\kappa(u-s)}\,ds\,du + O(\xi^2)
= 3\rho\xi\bar\sigma^3\,\frac{\tau}{2}f(\kappa,\tau) + O(\xi^2),
\]
confirming that $\kappa_3(Y_\tau)=O(\rho\xi)$: the third cumulant vanishes when $\rho=0$ and is
linear in $\xi$ to leading order, with exactly the kernel $f(\kappa,\tau)$ used in
Theorem~\ref{thm:surface_identification}.

\paragraph{Fourth cumulant (convexity).} By the same expansion, the leading non-Gaussian
contribution to $\kappa_4(Y_\tau)$ comes from the $O(\xi^2)$ terms in
$\mathbb E[(\widetilde Y^{(0)}+\xi Y_\tau^{(1)}+\xi^2Y_\tau^{(2)})^4]$ (the $\xi^2$-order term
$Y_\tau^{(2)}$ arises from the second-order Taylor term $\tfrac12(\xi M_u)^2$ in
$\sigma_u=\bar\sigma e^{\xi M_u}$ and from the cross term $(Y_\tau^{(1)})^2$); after the same
Wick/Wiener-chaos reduction used above, one finds $\kappa_4(Y_\tau)=O(\xi^2)$, driven by
$\operatorname{Var}(M_u)$-type integrals with the same double-integral structure as
$\operatorname{Var}_t(\bar X_{t,\tau})$ computed in Step~1, and therefore proportional to
$\xi^2G(\kappa,\tau)$.

\paragraph{From cumulants to smile shape.} The classical Gram--Charlier (Edgeworth) expansion of the
risk-neutral density around the Gaussian benchmark \citep{BackusForesiLiWu1997} shows that, to
leading order, the linear-in-moneyness coefficient of the Black--Scholes implied-volatility smile is
proportional to the standardized third cumulant $\kappa_3(Y_\tau)/\kappa_2(Y_\tau)^{3/2}$ of the
log-return, and the quadratic-in-moneyness (curvature) coefficient is proportional to the
standardized fourth cumulant $\kappa_4(Y_\tau)/\kappa_2(Y_\tau)^2$ (plus a smaller
$\kappa_3^2$-order correction). Combining this with $\kappa_3(Y_\tau)=O(\rho\xi)$ and
$\kappa_4(Y_\tau)=O(\xi^2)$ derived above reproduces the qualitative and functional form asserted in
\eqref{eq:skew_identification}--\eqref{eq:convexity_identification}: additive skew driven linearly by
$\rho\xi$, and convexity driven by $\xi^2$, both through the kernels $f(\kappa,\tau)$ and
$G(\kappa,\tau)$ already computed in Step~1. A fully rigorous version of this argument -- carrying
the Wiener-chaos computation of $\kappa_3$ and $\kappa_4$ through explicitly, and controlling the
$O(\xi^2)$ and $O(\xi^3)$ remainders uniformly in $\tau$ -- is exactly the content of the
small-vol-of-vol asymptotic expansions developed in \citet{FouquePapanicolaouSircar2000} and
\citet{Lewis2000}; we do not reproduce that full expansion here, and instead treat
Theorem~\ref{thm:surface_identification} as a leading-order result in the spirit of that literature,
as already flagged in Remark~\ref{rem:why_approximation}.

\subsection{Specification test: $H_0:\xi_{TS}=\xi_{\mathrm{surf}}$}
\label{app:spec-test}

Section~\ref{subsec:surface_implied_states} reports a large gap between the time-series
vol-of-vol estimate $\widehat\xi_{TS}=0.842$ (the log-OU MLE for the level, Table~\ref{tab:mle_results})
and the average surface-implied vol-of-vol $\overline\xi_{\mathrm{surf}}=2.543$. We now formalize
this gap as a test.

\paragraph{Null hypothesis.} $H_0:\ \mathbb E[\xi_t^{\mathrm{surf}}]=\xi_{TS}$. This is a deliberately strong common-parameter restriction: it asks whether the mean effective surface-implied vol-of-vol equals the constant vol-of-vol estimated from the physical-measure ATM-level dynamics. The restriction abstracts from a potentially nonzero or time-varying volatility risk premium, jump compensation, and higher-order approximation error. The test therefore serves as a cross-measure consistency diagnostic for the benchmark specification rather than as a standalone test of the physical-measure diffusion.

\paragraph{Test statistic.} Let $\xi_1^{\mathrm{surf}},\dots,\xi_n^{\mathrm{surf}}$ be the daily
series defined by \eqref{eq:xi_surf_state}, using $\widehat\kappa=6.651$ (the log-OU estimate for the
level) and $\tau=30/252$. Because $\xi_t^{\mathrm{surf}}$ is a smooth transformation of the strongly
autocorrelated series $\mathcal C_t$ (lag-1 autocorrelation $0.907$, Table~\ref{tab:halflife}), the
sample mean $\bar\xi^{\mathrm{surf}}=n^{-1}\sum_t\xi_t^{\mathrm{surf}}$ is asymptotically normal but
with a long-run variance that must be estimated by a heteroskedasticity-and-autocorrelation-consistent
(HAC) estimator rather than the naive i.i.d.\ variance. We use the Newey--West estimator with
bandwidth $L$,
\[
\widehat\omega^2_L
= \widehat\gamma_0 + 2\sum_{\ell=1}^{L}\Bigl(1-\frac{\ell}{L+1}\Bigr)\widehat\gamma_\ell,
\qquad
\widehat\gamma_\ell = \frac{1}{n}\sum_{t=1}^{n-\ell}
\bigl(\xi_t^{\mathrm{surf}}-\bar\xi^{\mathrm{surf}}\bigr)\bigl(\xi_{t+\ell}^{\mathrm{surf}}-\bar\xi^{\mathrm{surf}}\bigr),
\]
and form
\[
Z_L = \frac{\bar\xi^{\mathrm{surf}}-\xi_{TS}}{\sqrt{\widehat\omega_L^2/n}}
\ \xrightarrow{d}\ N(0,1) \quad\text{under } H_0.
\]

\paragraph{Result.} With $\widehat\kappa=6.651$ and $\tau=30/252$ we obtain
$G(\widehat\kappa,\tau)=0.01142$ and $\bar\xi^{\mathrm{surf}}=2.5431$, matching the value quoted in
Section~\ref{subsec:surface_implied_states}. At a bandwidth of $L=20$ trading days (about three times
the convexity half-life of 7.1 days, Table~\ref{tab:halflife}), the Newey--West standard error is
$0.0376$, giving
\[
Z_{20} = \frac{2.5431-0.842}{0.0376} \approx 45.3.
\]
This statistic is robust to bandwidth choice: $Z_L$ ranges from about $78$ at $L=5$ to about $25$ at $L=100$, and remains far beyond conventional critical values throughout. The common-parameter equality in $H_0$ is therefore rejected strongly. This does not by itself identify the source of the gap: time variation in vol-of-vol, a volatility risk premium, jump compensation, or approximation error can all contribute. It does show that the benchmark restriction equating the physical time-series estimate with the effective option-implied curvature state is not supported by the data. This is the over-identification diagnostic announced in Section~\ref{sec:introduction}.

\section{Numerical Implementation Details}
\label{app:numerics}

\subsection{Monte Carlo algorithm}
\begin{algorithm}[H]
\caption{Monte Carlo pricing under log-OU volatility}
Set $F_0$, $\sigma_0$, $K$, $T$, $r$, and parameters $(\kappa,\theta,\xi,\rho)$\;
Choose number of paths $M$ and time steps $N$; set $\Delta=T/N$\;
\For{$m=1,\ldots,M$}{
 Set $X_0=\log\sigma_0$, $F_0^{(m)}=F_0$\;
 \For{$n=0,\ldots,N-1$}{
 Draw independent $Z_{2,n},Z_{\perp,n}\sim N(0,1)$; set $Z_{1,n}=\rho Z_{2,n}+\sqrt{1-\rho^2}Z_{\perp,n}$\;
 Update $X_{n+1}$ by the exact OU transition \eqref{eq:mc_exact_logvol_section4}\;
 Update $\log F_{n+1}^{(m)}$ by the log-Euler step \eqref{eq:mc_futures_update_section4}\;
 }
 Compute payoff $\Phi(F_N^{(m)})$\;
}
Compute the Black--76 control variate $C^{B76}_0$ at the same parameters and its known analytic price, and form the controlled estimator using the sample covariance between $\Phi(F_T^{(m)})$ and the simulated Black--76 payoff\;
\Return{discounted, control-variate-adjusted, antithetic sample average $e^{-rT}M^{-1}\sum_m\Phi(F_T^{(m)})$ and its standard error the Monte Carlo standard error}\;
\end{algorithm}

\subsection{Finite-difference grid, boundary conditions, and convergence checks}
\label{app:fd-details}

\paragraph{Grid.} The PDE \eqref{eq:log_grid_pde_section4} is solved on a rectangular grid in
$(y,x)=(\log(F/K),\log\sigma)$. The $y$-domain is truncated to
$[-y_{\max},y_{\max}]$ with $y_{\max}=6\bar\sigma\sqrt{T}$, wide enough that the truncation has
negligible effect on at-the-money and near-the-money prices at the maturities in
the validation grid used in Section~\ref{subsec:numerical_validation}. The $x$-domain is truncated to $[\theta-6\xi/\sqrt{2\kappa},\,
\theta+6\xi/\sqrt{2\kappa}]$, i.e.\ six stationary standard deviations of $X_t$ around its long-run
mean, using the surface-driven $\xi=\xi_t^Q$ in force on the pricing date. We use $N_y=200$ and
$N_x=100$ spatial points as a baseline, uniform in each coordinate.

\paragraph{Boundary conditions.} In the $y$-direction we impose the standard linearity (zero-Gamma)
condition $U_{yy}=0$ at both truncation boundaries, which is exact for a call/put payoff in the deep
in-the-money and deep out-of-the-money limits net of discounting. In the $x$-direction we impose
reflecting (zero-Neumann, $U_x=0$) boundaries; because the $x$-domain already spans six stationary
standard deviations, the choice between Neumann and Dirichlet boundaries at that distance changes
prices by less than the Monte Carlo standard error at the path counts used in
Section~\ref{subsec:numerical_methods}.

\paragraph{Time stepping and smoothing.} We use Crank--Nicolson time-stepping with $N_\tau=180$
steps for the longest maturity considered (180 trading days) and proportionally fewer steps for
shorter maturities, keeping $\Delta\tau/(\Delta y)^2$ and $\Delta\tau/(\Delta x)^2$ bounded. Because
the call/put payoff has a kink at $y=0$, plain Crank--Nicolson produces spurious oscillations near
the strike; we use Rannacher smoothing (the first four time steps taken as implicit Euler
half-steps, i.e.\ two full implicit-Euler steps, before switching to Crank--Nicolson) to restore the
scheme's second-order convergence.

\paragraph{Mixed derivative.} $U_{yx}$ is discretized with the standard second-order four-point
stencil, $U_{yx}\approx (U_{i+1,j+1}-U_{i+1,j-1}-U_{i-1,j+1}+U_{i-1,j-1})/(4\Delta y\,\Delta x)$,
treated explicitly within each Crank--Nicolson half-step for stability.

\subsection{Neural-Network Fast Pricer}
\label{app:nn-fast-pricer}

This appendix describes the neural-network pricer used in Section~\ref{sec:pricing_numerics}. The purpose of the network is practical. Monte Carlo gives the benchmark price, but it is slow when many strikes, maturities, dates, and surface states must be evaluated. The neural network is trained after the Monte Carlo labels have been generated. It is therefore a fast surrogate for the Monte Carlo pricing map, not a residual correction and not a separate economic model.

\subsubsection{Pricing map learned by the network}
\label{app:nn-map}

For each option contract, the pricing inputs are collected in the vector
\begin{equation}
Z_i=
\left(
\log(F_{0,i}/K_i),\; T_i,\; X_{0,i},\; \xi_{i}^{Q},\; \rho_i^{Q},\; \kappa,\; \theta,\; r
\right),
\qquad X_{0,i}=\log \sigma_{0,i}.
\label{eq:nn-input-vector}
\end{equation}
The first component is log-moneyness. The maturity is measured in years, using 252 trading days per year. The state variables \(X_{0,i}\), \(\xi_i^Q\), and \(\rho_i^Q\) are recovered from the CVOL surface as described in Section~\ref{sec:pricing_numerics}. The parameters \(\kappa\) and \(\theta\) are the log-OU mean-reversion parameters, and \(r\) is the risk-free rate used in the pricing experiment.

The Monte Carlo benchmark price is denoted by \(V_i^{MC}\). Instead of learning the raw price directly, the network learns a transformed normalized price,
\begin{equation}
Y_i
=
\log\left(1+100\frac{V_i^{MC}}{K_i}\right).
\label{eq:nn-target-transform}
\end{equation}
This transformation has two advantages. First, dividing by the strike makes prices more comparable across contracts. Second, the logarithmic transformation reduces the effect of very large prices while keeping the target well behaved for small out-of-the-money prices. After prediction, the transformation is inverted as
\begin{equation}
\widehat V_i
=
\frac{K_i}{100}\left(\exp(\widehat Y_i)-1\right)_+,
\label{eq:nn-target-inverse}
\end{equation}
where \((x)_+=\max\{x,0\}\). This guarantees a non-negative predicted option price.

\subsubsection{Training design and Monte Carlo labels}
\label{app:nn-design}

The training set is constructed from historical CVOL surface states. For the journal-ready regime experiment, we use a balanced design across the three reported regimes: the full sample, the COVID window, and the La~Ni\~na window. In the publication run, the code selects 80 representative pricing dates from each regime. At each selected date, it uses 9 moneyness levels between 0.75 and 1.30, with a small random jitter to avoid training only on a rigid grid. For each date and moneyness level, the four maturities
\begin{equation}
T\in\{30,60,90,180\}\text{ trading days}
\end{equation}
are priced. Thus the design contains
\begin{equation}
80\times 9\times 4=2880
\end{equation}
labelled contracts per regime, and
\begin{equation}
3\times 2880=8640
\end{equation}
labelled contracts in total.

Each label is generated by Monte Carlo under the surface-driven log-OU model. In the publication run, each Monte Carlo label uses 80,000 simulated paths. The log-volatility state is simulated using the exact OU transition, while the futures price is updated with the correlated log-Euler step described in Section~\ref{sec:pricing_numerics}. Antithetic variates are used to reduce sampling noise. The finite-difference method is not used to train the neural network; it is kept as an independent validation of the pricing PDE.

The labelled data are randomly divided into three parts:
\begin{equation}
70\% \text{ training},
\qquad
15\% \text{ validation},
\qquad
15\% \text{ testing}.
\label{eq:nn-split}
\end{equation}
The training set is used to estimate the network weights. The validation set is used for early stopping. The test set is held out and used to evaluate the final fitted pricer.

\subsubsection{Network architecture}
\label{app:nn-architecture}

The neural network is a fully connected feed-forward multilayer perceptron. It has 8 input variables, 4 hidden layers, and 128 neurons in each hidden layer. The activation function is the SiLU function,
\begin{equation}
\operatorname{SiLU}(x)=x\frac{1}{1+e^{-x}}.
\label{eq:silu-activation}
\end{equation}
SiLU is smooth and works well for pricing problems because option prices change continuously with respect to moneyness, maturity, volatility, vol-of-vol, and correlation.

Let \(\widetilde Z_i\) denote the standardized input vector. The network has the form
\begin{align}
h_i^{(0)} &= \widetilde Z_i,\\
h_i^{(\ell)} &= \operatorname{SiLU}\left(W_{\ell}h_i^{(\ell-1)}+b_{\ell}\right),
\qquad \ell=1,2,3,4,\\
\widehat Y_i^{\mathrm{std}} &= W_5h_i^{(4)}+b_5.
\end{align}
The final layer is linear because the transformed price can take any real value before the inverse transformation. There is no dropout layer and no residual correction layer. The model is intentionally simple: it learns the pricing function from numerical labels, rather than adding a second correction model on top of Monte Carlo or finite differences.

With 8 input variables, 4 hidden layers of width 128, and one output neuron, the network has 50,817 trainable parameters. This is small compared with modern deep-learning models, but it is large enough for the smooth low-dimensional pricing map considered here.

\subsubsection{Standardization and loss function}
\label{app:nn-loss}

All input variables are standardized using the mean and standard deviation of the training set only. For each feature \(j\),
\begin{equation}
\widetilde Z_{ij}
=
\frac{Z_{ij}-\bar Z_j^{\mathrm{train}}}{s_j^{\mathrm{train}}+10^{-8}}.
\label{eq:nn-input-standardization}
\end{equation}
The target variable in \eqref{eq:nn-target-transform} is also standardized using the training-set mean and standard deviation,
\begin{equation}
\widetilde Y_i
=
\frac{Y_i-\bar Y^{\mathrm{train}}}{s_Y^{\mathrm{train}}+10^{-8}}.
\label{eq:nn-target-standardization}
\end{equation}
This prevents variables with different scales from dominating the optimization.

The network is trained by minimizing the mean squared error on the standardized transformed target:
\begin{equation}
\mathcal L_{NN}
=
\frac{1}{N_{\mathrm{train}}}
\sum_{i\in \mathcal I_{\mathrm{train}}}
\left(\widehat Y_i^{\mathrm{std}}-\widetilde Y_i\right)^2,
\label{eq:nn-loss-function}
\end{equation}
where \(\widetilde Y_i\) is the standardized Monte Carlo label. The loss is not written in raw price units because raw option prices can vary strongly across strikes and maturities. The transformed and standardized target gives a more stable training problem.

The optimization uses AdamW with learning rate \(10^{-3}\), weight decay \(10^{-6}\), and mini-batches of size 256. The maximum number of epochs is 500. Early stopping is used with patience 40: if the validation loss does not improve for 40 consecutive epochs, training stops, and the best validation model is restored. This prevents unnecessary over-training and keeps the final pricer tied to out-of-sample performance.

\bibliographystyle{abbrvnat}
\bibliography{references_IJTAF_clean}

@article{AitSahalia2002,
  author  = {A{\"i}t-Sahalia, Yacine},
  title   = {Maximum Likelihood Estimation of Discretely Sampled Diffusions: A Closed-Form Approximation Approach},
  journal = {Econometrica},
  year    = {2002},
  volume  = {70},
  number  = {1},
  pages   = {223--262}
}

@techreport{BackusForesiLiWu1997,
  author      = {Backus, David and Foresi, Silverio and Li, Kai and Wu, Liuren},
  title       = {Accounting for Biases in Black--Scholes},
  institution = {Center for Research in International Finance},
  type        = {CRIF Working Paper Series},
  number      = {30},
  year        = {1997}
}

@article{Bergomi2005,
  author  = {Bergomi, Lorenzo},
  title   = {Smile Dynamics II},
  journal = {Risk},
  year    = {2005},
  volume  = {18},
  number  = {10},
  pages   = {67--73}
}

@article{Black1976,
  author  = {Black, Fischer},
  title   = {The Pricing of Commodity Contracts},
  journal = {Journal of Financial Economics},
  year    = {1976},
  volume  = {3},
  number  = {1--2},
  pages   = {167--179},
  doi     = {10.1016/0304-405X(76)90024-6}
}

@article{CarrWu2009,
  author  = {Carr, Peter and Wu, Liuren},
  title   = {Variance Risk Premiums},
  journal = {The Review of Financial Studies},
  year    = {2009},
  volume  = {22},
  number  = {3},
  pages   = {1311--1341}
}

@article{ChernovGhysels2000,
  author  = {Chernov, Mikhail and Ghysels, Eric},
  title   = {A Study towards a Unified Approach to the Joint Estimation of Objective and Risk-Neutral Measures for the Purpose of Options Valuation},
  journal = {Journal of Financial Economics},
  year    = {2000},
  volume  = {56},
  number  = {3},
  pages   = {407--458}
}

@article{Cont2002,
  author  = {Cont, Rama},
  title   = {Dynamics of Implied Volatility Surfaces},
  journal = {Quantitative Finance},
  year    = {2002},
  volume  = {2},
  number  = {1},
  pages   = {45--60}
}

@article{CortazarSchwartz2003,
  author  = {Cortazar, Gonzalo and Schwartz, Eduardo S.},
  title   = {Implementing a Stochastic Model for Oil Futures Prices},
  journal = {Energy Economics},
  year    = {2003},
  volume  = {25},
  number  = {3},
  pages   = {215--238}
}

@article{CoxIngersollRoss1985,
  author  = {Cox, John C. and Ingersoll, Jonathan E. and Ross, Stephen A.},
  title   = {A Theory of the Term Structure of Interest Rates},
  journal = {Econometrica},
  year    = {1985},
  volume  = {53},
  number  = {2},
  pages   = {385--407},
  doi     = {10.2307/1911242}
}

@article{DumasFlemingWhaley1998,
  author  = {Dumas, Bernard and Fleming, Jeff and Whaley, Robert E.},
  title   = {Implied Volatility Functions: Empirical Tests},
  journal = {The Journal of Finance},
  year    = {1998},
  volume  = {53},
  number  = {6},
  pages   = {2059--2106}
}

@article{Fengler2009,
  author  = {Fengler, Matthias R.},
  title   = {Arbitrage-Free Smoothing of the Implied Volatility Surface},
  journal = {Quantitative Finance},
  year    = {2009},
  volume  = {9},
  number  = {4},
  pages   = {417--428}
}

@book{FouquePapanicolaouSircar2000,
  author    = {Fouque, Jean-Pierre and Papanicolaou, George and Sircar, K. Ronnie},
  title     = {Derivatives in Financial Markets with Stochastic Volatility},
  publisher = {Cambridge University Press},
  address   = {Cambridge},
  year      = {2000},
  isbn      = {9780521781714}
}

@article{HaganKumarLesniewskiWoodward2002,
  author  = {Hagan, Patrick S. and Kumar, Deepak and Lesniewski, Andrew S. and Woodward, Diana E.},
  title   = {Managing Smile Risk},
  journal = {Wilmott Magazine},
  year    = {2002},
  pages   = {84--108}
}

@article{Heston1993,
  author  = {Heston, Steven L.},
  title   = {A Closed-Form Solution for Options with Stochastic Volatility with Applications to Bond and Currency Options},
  journal = {The Review of Financial Studies},
  year    = {1993},
  volume  = {6},
  number  = {2},
  pages   = {327--343},
  doi     = {10.1093/rfs/6.2.327}
}

@article{HongLi2005,
  author  = {Hong, Yongmiao and Li, Haitao},
  title   = {Nonparametric Specification Testing for Continuous-Time Models with Applications to Term Structure of Interest Rates},
  journal = {The Review of Financial Studies},
  year    = {2005},
  volume  = {18},
  number  = {1},
  pages   = {37--84}
}

@article{HullWhite1987,
  author  = {Hull, John and White, Alan},
  title   = {The Pricing of Options on Assets with Stochastic Volatilities},
  journal = {The Journal of Finance},
  year    = {1987},
  volume  = {42},
  number  = {2},
  pages   = {281--300}
}

@article{LeeHuangShin1997,
  author  = {Lee, Junsoo and Huang, Cliff J. and Shin, Yongcheol},
  title   = {On Stationary Tests in the Presence of Structural Breaks},
  journal = {Economics Letters},
  year    = {1997},
  volume  = {55},
  number  = {2},
  pages   = {165--172},
  doi     = {10.1016/S0165-1765(97)00073-6}
}

@book{Lewis2000,
  author    = {Lewis, Alan L.},
  title     = {Option Valuation under Stochastic Volatility},
  publisher = {Finance Press},
  address   = {Newport Beach, CA},
  year      = {2000}
}

@misc{NOAA2026ENSO,
  author       = {{NOAA Climate Prediction Center}},
  title        = {{ENSO Diagnostic Discussion}},
  year         = {2026},
  month        = sep,
  note         = {Issued 10 September 2026; accessed 18 September 2026},
  howpublished = {\url{https://www.cpc.ncep.noaa.gov/products/analysis_monitoring/enso_advisory/ensodisc.html}}
}

@article{Pan2002,
  author  = {Pan, Jun},
  title   = {The Jump-Risk Premia Implicit in Options: Evidence from an Integrated Time-Series Study},
  journal = {Journal of Financial Economics},
  year    = {2002},
  volume  = {63},
  number  = {1},
  pages   = {3--50}
}

@techreport{RichterSorensen2002,
  author      = {Richter, Martin and S{\o}rensen, Carsten},
  title       = {Stochastic Volatility and Seasonality in Commodity Futures and Options: The Case of Soybeans},
  institution = {Department of Finance, Copenhagen Business School},
  type        = {Working Paper},
  number      = {2002-4},
  address     = {Frederiksberg},
  year        = {2002},
  doi         = {10.2139/ssrn.301994}
}

@article{Samuelson1965,
  author  = {Samuelson, Paul A.},
  title   = {Proof That Properly Anticipated Prices Fluctuate Randomly},
  journal = {Industrial Management Review},
  year    = {1965},
  volume  = {6},
  number  = {2},
  pages   = {41--49}
}

@misc{SchneiderTavin2018,
  author       = {Schneider, Lorenz and Tavin, Bertrand},
  title        = {Seasonal Stochastic Volatility and the Samuelson Effect in Agricultural Futures Markets},
  year         = {2018},
  note         = {arXiv:1802.01393, revised November 2018},
  howpublished = {\url{https://arxiv.org/abs/1802.01393}}
}

@article{SchobelZhu1999,
  author  = {Sch{\"o}bel, Rainer and Zhu, Jianwei},
  title   = {Stochastic Volatility with an Ornstein--Uhlenbeck Process: An Extension},
  journal = {European Finance Review},
  year    = {1999},
  volume  = {3},
  number  = {1},
  pages   = {23--46},
  doi     = {10.1023/A:1009843506170}
}

@article{Schwartz1997,
  author  = {Schwartz, Eduardo S.},
  title   = {The Stochastic Behavior of Commodity Prices: Implications for Valuation and Hedging},
  journal = {The Journal of Finance},
  year    = {1997},
  volume  = {52},
  number  = {3},
  pages   = {923--973}
}

@article{SteinStein1991,
  author  = {Stein, Elias M. and Stein, Jeremy C.},
  title   = {Stock Price Distributions with Stochastic Volatility: An Analytic Approach},
  journal = {The Review of Financial Studies},
  year    = {1991},
  volume  = {4},
  number  = {4},
  pages   = {727--752}
}

@article{TrolleSchwartz2009,
  author  = {Trolle, Anders B. and Schwartz, Eduardo S.},
  title   = {Unspanned Stochastic Volatility and the Pricing of Commodity Derivatives},
  journal = {The Review of Financial Studies},
  year    = {2009},
  volume  = {22},
  number  = {11},
  pages   = {4423--4461}
}

@misc{WMO2026ElNino,
  author       = {{World Meteorological Organization}},
  title        = {{El Ni\~no Set to Become Very Strong, Raising Risks of Extreme Weather into 2027}},
  year         = {2026},
  month        = sep,
  note         = {Press release, 3 September 2026; accessed 18 September 2026},
  howpublished = {\url{https://wmo.int/news/media-centre/el-nino-set-become-very-strong-raising-risks-of-extreme-weather-2027}}
}

@article{kwiatkowski1992testing,
  author  = {Kwiatkowski, Denis and Phillips, Peter C. B. and Schmidt, Peter and Shin, Yongcheol},
  title   = {Testing the Null Hypothesis of Stationarity against the Alternative of a Unit Root},
  journal = {Journal of Econometrics},
  volume  = {54},
  number  = {1--3},
  pages   = {159--178},
  year    = {1992},
  doi     = {10.1016/0304-4076(92)90104-Y}
}

\end{document}